\documentclass[11pt]{article}
\usepackage{amsmath,amsfonts,amsthm,amssymb, graphicx,longtable}
\usepackage[hmargin={1.0in,1.0in},vmargin={1.0in,1.0in}]{geometry}
\usepackage{setspace}
\usepackage{subfig}
\usepackage[numbers]{natbib}

\theoremstyle{plain}
\newtheorem{thm}{Theorem}
\newtheorem{proposition}{Proposition}

\newtheorem{assumption}{Assumption}

\begin{document}

\title{Continuous-time Modeling of Bid-Ask Spread and Price Dynamics in
Limit Order Books }
\author{Jose Blanchet and Xinyun Chen}
\maketitle

\begin{abstract}
We derive a continuous time model for the joint evolution of the mid price
and the bid-ask spread from a multiscale analysis of the whole limit order
book (LOB) dynamics. We model the LOB as a multiclass queueing system and
perform our asymptotic analysis using stylized features observed
empirically. We argue that in the asymptotic regime supported by empirical
observations the mid price and bid-ask-spread can be described using only
certain parameters of the book (not the whole book itself). Our limit
process is characterized by reflecting behavior and state-dependent jumps.
Our analysis allows to explain certain characteristics observed in practice
such as: the connection between power-law decaying tails in the volumes of
the order book and the returns, as well as statistical properties of the long-run spread distribution.
\end{abstract}

\section{Introduction}

Limit order book (LOB) models have recently attracted a lot of attention in
the literature given their importance in modern financial markets and they
are used as trading protocol in most exchanges around the world. For a brief
review of worldwide financial markets that use LOB mechanism, see the first
paragraph of \citet{Gould&etal_2012}. As we shall discuss, the literature on
price modeling based on LOB dynamics has mostly focused on one side of the
order book, or price dynamics that are not fully informed by the LOB.

One of our main contributions is the construction of a continuous time model
for the joint evolution of the mid price and the bid-ask spread (see Theorem %
\ref{Thm_Main} in Section \ref{Section_CT_Model}). Such construction is
informed by the full LOB dynamics, which we model as a multiclass queueing
system (see Section \ref{Section_Building_Block}). We endow the multiclass
queue with characteristics that are inspired by common stylized features
which are observed empirically in order book data, such as: very fast speed
of orders relative to price changes, high cancellation rates, and power-law
tails (see for example Sections \ref{Section_Observations_Distributions} and %
\ref{Section_Empirical_Validation}). Some of these stylized features allow
us to justify the use of certain asymptotic limits and weak convergence
analyses which are applied to the LOB and ultimately give rise to our
continuous time pricing model.

Another contribution that it is important to highlight is that our analysis
sheds light on the connection between power-law tails which are present both
in the distribution of orders inside the book, and also in the realized
return distributions of price processes. The connection between these
features, which are documented in the statistical literature (see %
\citet{Bouchard&etal_2002}) are explained as a result of the statements
obtained in Theorem \ref{THM stochastic averaging} and Proposition \ref%
{Result_Hazard_Rate} in Section \ref{Section_CT_Model}. Basically the
power-law tails arising from the distribution of returns in the price
processes can be explained as a consequence of the power-law tails in the
distribution of orders inside the book, the effect of cancellation policies,
and the asymptotic regime under which LOBs operate.

We establish a one-to-one correspondence between the distribution of orders
inside the book and the price-return distribution assuming a specific form
of the cancellation policy (see Proposition \ref{Result_Hazard_Rate} for the
one-to-one correspondence and Assumption \ref{ass:cancelation} for the form
of the cancellation policy). We argue that such cancellation policy has
qualitative features observed in practice. For example, we postulate higher
cancellation rates for orders that are placed closer to the spread and lower
cancellation rates for orders placed away from the spread (see the
discussion at the end of Section \ref{Section_Conect_LOB_Price}). Further,
we also argue that the cancellation rate that we postulate is such that, in
statistical equilibrium, the probability at which a given order is executed
before cancellation is roughly the same regardless of where the order is
placed in book.

Although we believe that our cancellation policy is reasonable in some
circumstances, more generally, our results are certainly useful to obtain
insights into the form of cancellation used by market participants. This
insight could be obtained by comparing the distribution of orders inside the
book and the price return distribution relative to equation (\ref%
{eq:formula_theta}), which also contains the cancellation rates.

Furthermore, we also argue in Section \ref{Sec_Con_Others} that a strong
connection between distribution of orders inside the book and price-return
distribution is to be expected in a different asymptotic regime, namely,
that in which market and limit orders arrive at comparable speed, and
constant cancellation rates per order across the book. We thus believe that
our analysis provides significant evidence for such connection between these
distributions.

We envisage our model to be useful in intra-day trading. Underlying the
motivation behind the construction of our model is our belief that there is
a significant amount of information in the order book which can be used to
help describe the evolution of the price and bid-ask spread in the order of
a few hours. At the same time, we recognize that it might be challenging in
practice to keep track of the full LOB to describe price dynamics.
Fortunately, as we demonstrate here, under the asymptotic regime that we
utilize, implied by empirical observations, it is possible to keep track of
the prices in continuous time using only a two dimensional Markov process.
Our continuous-time pricing model can be calibrated, for example, from the
historical data of the distribution of limit orders inside the book, and
then fine tuned (through an additional parameter which we call the patience
ratio, $c_{p}$) again based on historical price return distribution data
(see for example the discussion in Section \ref{Section_Empirical_Validation}%
\ involving Empirical Observation 3). Therefore, we are able to use the
information on the order book in a meaningful, yet relatively simple, way to
inform the future evolution of prices. Moreover, as we shall illustrate,
using simulated data, our final model captures empirical features observed
in practice (see Section \ref{Section_Empirical_Validation}).

Let us discuss briefly the elements that distinguish our work from previous
contributions. As we indicated earlier, we built our model from a multiclass
queueing system. Queueing theory provides a natural environment for the
study of LOB dynamics at a microstructure level. Consequently, it is no
surprise that there is a fast growing literature that leverages off the use
of queueing theory in order to analyze LOBs and the corresponding price
dynamics. \citet{Cont&Stoikov&Talreja_2010} introduces birth-death queueing
model similar to our prelimit model discussed in Section \ref%
{Section_Building_Block}. They demonstrate that this model can lead to
computable conditional probabilities of interest in terms of Laplace
transforms studied in queueing theory. A scaling is introduced in a
subsequent paper \citet{Cont&Larrard_2010} in which a continuous-time
process is derived for price dynamics, but there the authors assume that
price dynamics only depend on best bid and ask quotes. In contrast, we
derive a two dimensional Markovian process for the best bid and best ask
quotes from the whole order book model. As a result, we arrive at a limiting
process which is different from that of \citet{Cont&Larrard_2010}. In %
\citet{Maglaras&etal_2012}, queueing models are used to address the problem
of routing of orders in a fragmented market with different books. More
recently, \citet{Lakner&etal_2013}. discuss a one sided order book and track
the whole process using measure-valued characteristics. Although they
analyze the system in a high-frequency trading environment their scaling
does not appear to highlight the role of the cancellations relative to what
occurs in the markets, in which a large proportion of the orders are
actually cancelled. In contrast we not only consider the two sides of the
book, but we believe our scalings better preserve the empirical features
observed in practice.

As mentioned earlier, we use multiscale analysis, which allows us to replace
most of the stochasticity in the book by steady-state dynamics. The paper by %
\citet{Sowers&etal_2012} also takes advantage of multiscale analysis, but
their model is not purely derived from the microstructure characteristics at
the level of arrivals of limit and market orders and therefore the ultimate
model is different from the one we obtain. A recent paper by %
\citet{Horst&Paulsen_2013} provides a law of large numbers description of
the order book and the price process, which is in the end deterministic and
therefore also different in nature to the stochastic model we derive here.
Nevertheless, we feel that the spirit of \citet{Horst&Paulsen_2013} is close
to the work we do here.

A related literature on model building for order book dynamics relates to
the use of self-exciting point processes. The approach is somewhat related
to the queueing perspective, although the models are more aggregated than
the ones we consider in this paper; for example, in \citet{Zheng&etal_2013}
a model is discussed that considers only best bid and ask quotes but with a
self-exciting mechanism with constraints (see also \citet{Cartea&etal_2011}, %
\citet{Toke_2010} and the references therein) for more information on
self-exciting processes in high-frequency trading.

This manuscript is organized as follows. In Section \ref%
{Section_Building_Block} we discuss the pre-limit model underlying the LOB
dynamics. In Section \ref{Section_Observations_Distributions} we discuss
some empirical observations that inform the construction of certain
approximations in our model which in particular allow to connect the price
increments and the distribution of orders in the LOB. In Section \ref%
{Section_CT_Model} we present our asymptotic scalings and our continuous
time price model. Finally, in Section \ref{Section_Empirical_Validation} we
discuss how, using simulated data, our final model captures empirical
features observed in practice.

\section{Basic Building Blocks \label{Section_Building_Block}}

Ultimately, our goal is to construct a continuous time model for the joint
evolution of the mid price and the bid-ask spread, which is informed by the
whole order book dynamics in such a way that key stylized features are
captured. In the end our model will be obtained as an asymptotic limit which
is informed by stylized features observed empirically. We first discuss the
building blocks of our model in the prelimit.

The building blocks of our model are consistent with prevalent limit order
book models that describe the interactions between order flows, market
liquidity and price dynamics, such as in \citet{Bouchard&etal_2002} %
\citet{Cont&Stoikov&Talreja_2010}. In most existing models, the arrival rate
of limit orders corresponding, say, to a given price is given as a function
of the distance between such given price and the best price of all present
limit orders of the same type (buy or sell).

The best price of all limit sell (and buy) orders is called the ask (and
bid) price and is usually denoted by $a(t)$ (and $b(t)$) at time $t$.
However, as observed in recent empirical data, due to the growing popularity
of algorithmic trading, limit orders are put and canceled without being
executed at high frequency, especially at positions between the best ask or
bid prices (fleeting orders). Therefore, the continuous observation of the
best bid-ask prices may result in a process with too much \textquotedblleft
noise\textquotedblright\ due to variability caused by cancellations of such
fleeting orders.

Instead, we shall construct our continuous time model by looking at the
prices only at time at which an actual trade occurs; we call these
quantities \textit{prices} \emph{per trade}. We believe that this is the
natural time scale at which track the evolution of the LOB in order to
derive a continuous time price process.

The relation between the \textit{price-per-trade process} $(\bar{a}(\cdot ),%
\bar{b}(\cdot ))$ and $(a(\cdot ),b(\cdot ))$ is as follows. Suppose $%
\{t_{k}:k\geq 1\}$ are the arrival times of market orders (on both sides),
then $(\bar{a}(t),\bar{b}(t))=(a(t_{k}),b(t_{k}))$ if $t_{k}\leq t<t_{k+1}$.
As indicated, intuitively, we can think of $(\bar{a}(\cdot ),\bar{b}(\cdot
)) $ as a mechanism to filter out the noise made by fleeting orders in the
prelimit process $(a(\cdot ),b(\cdot ))$. The model that we consider in the
prelimit is described as follows.

\bigskip

\textbf{Model Dynamics and Notation:}

\begin{enumerate}
\item Limit orders or market orders arrive one at a time (i.e. there are no
batch arrivals).

\item Arrivals of limit buy orders and of sell orders are modelled as two
independent Poisson processes with equal rate $\lambda $.

\item Arrivals of market buy orders and of sell orders are modelled as two
independent Poisson processes with equal rate $\mu $.

\item Let $\{t_{k}:k\in \mathbb{Z^{+}}\}$ be the arrival times of the market
orders (either buy or sell, so these are the arrivals of a Poisson process
with rate $2\mu $).

\item The prices take values on the lattice $\{i\delta :i\in \mathbb{Z^{+}}%
\} $ and we observe their change at time lattice points $\{t_{k}:k\in
\mathbb{Z^{+}}\}$. The parameter $\delta $ is called the tick size and in
the sequel we will specify an asymptotic relationship between $\delta $ and
the frequency of arrival times of orders.

\item At time $t$, the ask price $a(t)$ (the bid price $b(t)$) equals the
minimum (maximum) of prices of all limit sell (buy) orders on the order book
at time $t$.

\item For $t_{k}\leq t<t_{k+1}$, an order placed at a relative ask price
equal to $i\delta $ at time $t$ implies that the order is posted for ask
price equal to $a\left( t_{k}\right) +i\delta $. Similarly, a relative bid
price equal to $i\delta $ at time $t$ implies an absolute bid price equal to
$b\left( t_{k}\right) -i\delta $. %We allow $i\leq
%0$ as long as $i\delta >-(a\left( t\right) -b\left( t\right) )/2$ -- in
%order to avoid overlap.

\item For $t_{k}\leq t<t_{k+1}$, upon its arrival at time $t$, a limit buy
(and sell) order sits at a relative price equal to $i\delta $ with
probability $p(i\delta ;\bar{a}(t_{k}),\bar{b}(t_{k}))$. In particular, $%
p(i\delta ;a,b)=0$ for all $i\delta \geq -(a-b)/2$ so that the incoming
limit buy and sell orders do not overlap with each other.

\item An order that right after time $t_{k}$ sits at a relative (buy or
sell)\ price equal to $i\delta $ is cancelled at rate $\alpha (i\delta ;\bar{%
a}(t_{k}),\bar{b}(t_{k}))$ during the time interval $[t_{k},t_{k+1})$.

\item A market order immediately transacts with any of the best matching
limit orders in the order book upon its arrival.
\end{enumerate}

\bigskip

\textbf{Remark:} We can actually weaken the assumption described in item 1.
above to allow market orders arriving in batches as long as the size of
incoming market order is less than the volume of standing limit orders at
the best quote.\bigskip

In our model, the ask and bid sides are two separate multi-class
single-server queues with exponentially distributed times between transition
of events (i.e. Markovian queues). On each side, the limit orders can be
view as customers that are divided into different classes according to their
relative prices (i.e. number of ticks to the best price). The class with
lower relative price has higher priority. The market orders play the role of
the server, as each of them causes a departure of a limit order from the
best tick price. In other words, limit orders are served at the same rate as
the arrival rate of market orders and the market orders pick customers from
the non-empty class with the highest priority.

It is important to note that between the arrivals of two consequent market
orders, the dynamic of the limit order book is equivalent to a set of
independent infinite-server systems. One such infinite-server system for
each class in each of the sides (buy and sell) of the order book. The
\textquotedblleft service rate\textquotedblright\ of each such infinite
server system is equal to the cancellation rate of the corresponding class.

We now proceed to develop the main ingredients of our model using stylized
features that are prevalent in market data.

\section{Empirical Observations, Price, and LOB's Distributions\label%
{Section_Observations_Distributions}}

We now discuss several empirical features that motivate the asymptotic
regime that we consider. In particular, these observations will help us
inform the asymptotic distribution of the price increments in intermediate
time scales (order of several seconds).

\subsection{Empirical Observations and Distribution of Price Increments}

\textbf{Empirical Observation 1: Multi-Scale Evolution of Limit Order Flows
and the Occurrence of Trades .} Table 1 is a sample from the descriptive
statistics of TAQ data from \citet{Cont&Kukanov&Stoikov_2013}. In
particular, we would like to hight the contrast between the daily number of
updates, which include the submission, cancellation and transactions of
limit orders at the best quote, and the daily number of trades
(transactions). Since each market order causes a transaction of limit
orders, the fact that the daily number of updates of limit orders at the
best quote is much bigger than that of trades indicate that the evolution of
limit orders is much more frequent than the arrivals of market orders in the
limit order book. Moreover, such difference is prevalent in both
high-liquidity stocks (such as Bank of America) and low-liquidity stocks
(such as CME Group). \emph{We will adopt the relative fast speed of the
limit order flows relative to market order flow (or occurrences of trade) in
our asymptotic regime.}
%Table 1 from \citet{Cont&Larrard_2010} consists of market data from a
%particular trading day. This table illustrates that the speed at which
%orders arrive to the order book is much faster than the speed at which the
%price changes. One can see that thousands of events in the LOB, for instance
%the arrivals of new orders, occur in just a few seconds. Meanwhile, the
%prices (bid and ask) change only thousands of times after the whole trading
%day -- about 6.5 hours. \textit{We will adopt the relative fast speed of the
%order book relative to the price changes.}

\begin{table}[h]
\caption{Daily average of 50 randomly chosen stocks in NYSE over 21 trading
days in April 2010.}\center
\begin{tabular}{|c|c|c|}
\hline
& Daily Number of best quote updates & Daily number of trades \\ \hline
Bank of America & $1529395$ & $15008$ \\ \hline
CME Group & $38504$ & $1412$ \\ \hline
Grand mean & $223132$ & $4552$ \\ \hline
\end{tabular}%
\end{table}
%\begin{table}[h]
%\caption{Market Data}\center
%\begin{tabular}{|c|c|c|}
%\hline
%(June 28th 2008) & Average number of orders in 10s & Number of price changes
%in whole day \\ \hline
%Citigroup & 4469 & 12499 \\ \hline
%General Electric & 2356 & 7862 \\ \hline
%General Motors & 1275 & 9016 \\ \hline
%\end{tabular}%
%\end{table}
%\smallskip
%The huge difference A more detailed study on the amount of limit and market orders on the bid/ask tick
\bigskip

\textbf{Empirical Observation 2: Fleeting Orders and High Cancellation Rate.
}

Due to the prevalence of algorithmic trading in these days, the cancellation
rate of limit orders has increased dramatically over recent years. Among
recent studies on the Nasdaq INET data, \citet{Hasbrouck&Saar_2009} compares
the data from year 1990, 1999 and 2004 and find that the cancellation rate
has increased dramatically, while \citet{Hautsch&Huang_2011} finds that
about $95\%$ of the limit orders are canceled in the 2010 data as shown in
Table 2.
\begin{table}[h]
\caption{Percentage of limit orders that are canceled without (partial)
execution for 10 stocks on NASDAQ. Samples are collected in October 2010,
covering 21 trading days. }\center
\vspace{.3 cm}
\begin{tabular}{cccccccccc}
\hline
GOOG & ADBE & VRTX & WFMI & WCRX & DISH & UTHR & LKQX & PTEN & STRA \\ \hline
97.52 & 92.57 & 93.82 & 95.25 & 92.83 & 94.56 & 95.54 & 96.62 & 91.62 & 95.57
\\ \hline
\end{tabular}%
\end{table}
As suggested by this empirical work, the high cancellation rate can be
attributed to the large proportion of \textquotedblleft
fleeting\textquotedblright\ limit orders which are usually put inside the
spread and then canceled immediately if not executed. For example, %
\citet{Hautsch&Huang_2011} reports that in 2010 Nasdaq data the mean
inter-arrival time of market orders is about 7 seconds while the mean
cancellation rate of limit orders inside the spread is less than 0.2
seconds. However, limit orders deep outside the spread are more patient.
\textit{We will assume that the arrival rates and cancellation rates are
substantially higher than the arrival rates of market orders in our model.}

\bigskip

Recall that $\{t_{k}:k\geq 1\}$ is the sequence of arrival times of market
orders (on both sides). We now discuss the distribution of the ask
price-per-trade increment $(\bar{a}(t_{k+1})-\bar{a}(t_{k}))$. Similar
results on the distribution of the bid price-per-trade increment $(\bar{b}%
(t_{k+1})-\bar{b}(t_{k}))$ follow by symmetry.

Note that the following identity of events holds
\begin{equation*}
\{\bar{a}(t_{k+1})-\bar{a}(t_{k})\geq i\delta \}=\{\text{The queues at
relative prices lower than }i\delta \text{ are all empty at time }t_{k+1}\}.
\end{equation*}%
Let
\begin{equation*}
\theta (i\delta ;a,b)=P_{\pi }(\text{the queues at relative tick prices
lower then }i\delta \text{ are all empty}),
\end{equation*}%
where $\pi $ is the stationary distribution of the underlying infinite
server queues associated to each class (which are all independent). The $i$%
-th queue has arrival rate $\lambda p(i\delta ;\bar{a}(t_{k}),\bar{b}(t_{k}))
$ and cancellation rate $\alpha (i\delta ;\bar{a}(t_{k}),\bar{b}(t_{k}))$.
It is known that the stationary distribution of an infinite server queue
with arrival rate $\lambda p(i\delta ;\bar{a}(t_{k}),\bar{b}(t_{k}))$ and
service rate $\alpha (i\delta ;\bar{a}(t_{k}),\bar{b}(t_{k}))$ is Poisson
with parameter $\lambda p(i\delta ;\bar{a}(t_{k}),\bar{b}(t_{k}))/\alpha (i;%
\bar{a}(t_{k}),\bar{b}(t_{k}))$, therefore we have that
\begin{equation}
\theta (i\delta ;\bar{a}(t_{k}),\bar{b}(t_{k}))=\exp \left( ~~-\sum_{j\delta
\geq -(\bar{a}(t_{k})-\bar{b}(t_{k}))/2}^{i-1}\frac{\lambda p(j\delta ;\bar{a%
}(t_{k}),\bar{b}(t_{k}))}{\alpha (j\delta ;\bar{a}(t_{k}),\bar{b}(t_{k}))}%
\right) .  \label{eq:formula_theta}
\end{equation}

As $\lambda>>\mu $ is large, suggested by Empirical Observation 1, and $%
\alpha (\cdot )>>\mu $ as suggested by the Empirical Observation 2, we can
typically approximate the distribution of the queue lengths at time $t_{k+1}$
(given the state of the system at time $t_{k}$) by the associated
steady-state distribution of the queues. More precisely, we expect the
approximation $P(\bar{a}(t_{k+1})-\bar{a}(t_{k})\geq i\delta |\bar{a}(t_{k}),%
\bar{b}\left( t_{k}\right) )\approx \theta (i\delta ;\bar{a}(t_{k}),\bar{b}%
\left( t_{k}\right) )$ to hold. This heuristic is made rigorous in the
following theorem, the proof, which is given in Section \ref%
{App_Sect_Proof_THM stochastic averaging}, is based on the so-called
stochastic averaging principle, see \citet{Kurtz_1992}. We use $D[0,\infty )$
to denote the space of right continuous with left limits functions from $%
[0,\infty )$ to $\mathbb{R}$ endowed with the Skorokhod $J_{1}$ topology
(see \citet{Billingsley_1999} for reference).

\begin{thm}
\label{THM stochastic averaging}Consider a sequence of LOB systems indexed
by $n$. In the $n$-th system, the total number of orders in the order book
is given by $q_{n}=q<\infty $, the distribution of these orders in the order
book is assumed to satisfy $\bar{a}_{n}\left( 0\right) =a_{n}\left( 0\right)
=\bar{a}$ and $\bar{b}_{n}(0)=b_{n}\left( 0\right) =\bar{b}$. We assume that
the arrival rate of market orders satisfies $\mu _{n}=\mu $ and the
distribution of incoming limit orders is $p_{n}(\cdot ;\cdot ,\cdot
)=p(\cdot ;\cdot ,\cdot )$ (i.e. constant along the sequence of systems).
Suppose there exists a sequence of positive number $\{\xi _{n}:n\geq 1\}$
such that $\lambda _{n}=\xi _{n}\lambda $, $\alpha _{n}(\cdot )=\xi
_{n}\alpha (\cdot )$ and $\xi _{n}\rightarrow \infty $ as $n\rightarrow
\infty $. We also assume the regularity condition that for any $(a,b)\in
\mathbb{Z}^{2}$
\begin{equation}
\lim_{i\rightarrow \infty }\theta (i\delta ;a,b)=0.  \label{Eq:Reg_Cond_Thm1}
\end{equation}%
Then the corresponding price process $(\bar{a}_{n}(\cdot ),\bar{b}_{n}(\cdot
))$ converges weakly in $D[0,\infty )$ to a pure jump process $(\hat{a}%
(\cdot ),\hat{b}(\cdot ))$. The process $(\hat{a}(\cdot ),\hat{b}(\cdot ))$
jumps at time corresponding to the arrivals of a Poisson process with rate $%
2\mu $. The size of the jumps are not independent, in particular, if $t$ is
a jump time, then the jump size at $t$ is a random vector following the the
distribution
\begin{eqnarray*}
&&P\left( \hat{a}(t)-\hat{a}(t-)=i\delta ,\hat{b}(t)-\hat{b}(t-)=j\delta ~|%
\hat{a}(t-),\hat{b}(t-)\right) \\
&=&[\theta (i\delta ;\hat{a}(t-),\hat{b}(t-))-\theta ((i+1)\delta ;\hat{a}%
(t-),\hat{b}(t-))]\times \lbrack \theta (j\delta ;\hat{a}(t-),\hat{b}%
(t-))-\theta ((j+1)\delta ;\hat{a}(t-),\hat{b}(t-))].
\end{eqnarray*}
\end{thm}

\textbf{Remark: }The regularity condition (\ref{Eq:Reg_Cond_Thm1}) not only
is quite natural, but it can be easily verified in terms of $p\left( i\delta
;a,b\right) $ and $\alpha \left( i\delta ;a,b\right) $ because there is an
explicit formula for $\theta \left( i\delta ;a,b\right) $ given below in
equation (\ref{eq:formula_theta}).

\bigskip

It is important to note that in the previous result we have held the arrival
rates of market orders constant, so this result simply describes the price
processes at times scales corresponding to the inter-arrival times of market
orders (i.e. in the order of a few seconds according to the representative
date discussed earlier). In Section \ref{Section_CT_Model} we shall
introduce a scaling that will allow us to consider the process in longer
time scales (several minutes or longer) by increasing the arrival rate of
market orders.

Theorem \ref{THM stochastic averaging} can be extended without much
complications to include more complex dynamics in the arrivals of the market
orders. For instance, one way to extend our model is to allow traders to
post market orders depending on the current bid-ask price. This modification
can be introduced as thinning procedure to the original Poisson process with
rate $2\mu $, the thinning parameter might depend on the observed bid-ask
price $(a(\cdot ),b(\cdot ))$. Other examples of the interactions between
market participants that can be included in our model extensions are
correlation between the buying and selling sides and dependence between
arrival rate of market orders and the spread width. \bigskip

\subsection{Connecting Distribution of Price Increments and LOB's
Distributions\label{Section_Conect_LOB_Price}}

We now will argue how Theorem \ref{THM stochastic averaging} allows us to
provide a direct connection between the increment distribution of the
price-per-trade process and the distribution of orders in the LOB. We
believe that this connection, although simple, is quite remarkable because
it forms the basis behind our idea of using information in the LOB to
predict the evolution of prices.

Let us assume the following form for the cancellation rate.

\begin{assumption}
\label{ass:cancelation}
\begin{equation}
\alpha (i\delta ;\bar{a}(t_{k}),\bar{b}(t_{k}))=\frac{\lambda}{c_p}\cdot%
\frac{p(i\delta;\bar{a}(t_k),\bar{b}(t_k))}{log(1-\sum_{j\leq i-1}p(j\delta;%
\bar{a}(t_k),\bar{b}(t_k)))-log(1-\sum_{j\leq i}p(j\delta;\bar{a}(t_k),\bar{b%
}(t_k)))},  \label{cancelation rate}
\end{equation}%
where $c_{p}>0$ is a constant we call the patience ratio of limit orders and
we will see in a moment that it plays an important role in the connection
between the distributions of limit order flow and price returns.
\end{assumption}

\bigskip

Under this assumption, by simple algebra in (\ref{eq:formula_theta}), we
have the following result.

\begin{proposition}
\label{Result_Hazard_Rate}
\begin{equation}  \label{eq:expression_theta}
\theta(i\delta;\bar{a}(t_k),\bar{b}(t_k))=(1-\sum_{j\leq i-1}p(j\delta;\bar{a%
}(t_k),\bar{b}(t_k)))^{c_p},
\end{equation}
\end{proposition}

Since $\theta(\cdot;a,b)$ is the tail of price return and $%
1-\sum_{j\leq\cdot}p(j;a,b)$ is the tail of the relative price of incoming
limit orders. Proposition \ref{Result_Hazard_Rate} indicates that the price
return inherits some statistical properties of the distribution of limit
order flow on the order book. In real market data, power-law tails are
reported in both the relative price of limit orders and the mid-price return
as discussed in the next subsection. In our model, given that the
distribution of the limit order flow $p(\cdot ;\bar{a},\bar{b})$ has a
power-law with exponent $v$, in other words, $\bar{F}^{c_{p}}(i\delta
)=\sum_{j\geq i}p(j\delta ;\bar{a},\bar{b})\approx (c_{1}i\delta )^{-v}$.
Then, as a direct consequence of (\ref{eq:expression_theta}), we have $P(%
\hat{a}(t)-\hat{a}(t-)\geq i\delta |\bar{a},\bar{b})\approx (c_{1}i\delta
)^{-c_{p}v}$ and therefore the price returns also follows a power law but
with different exponent $c_{p}v$.

For $c_{p}>1$, our model recovers an interesting phenomenon in real market
that the price return has a thinner tail than the relative price of limit
orders as reported in \citet{Bouchard&etal_2002}. Besides, in our model when
the limit orders are more patient ($c_{p}$ increases), the price return has
a thinner tail. This is consistent with the fact that price volatility
decreases when the market has more liquidity supply (since the limit orders
stand longer in the order book as $c_{p}$ increases).

We shall put some remarks on our Assumption \ref{ass:cancelation} on the
cancellation rate of limit orders. Although this assumption is made for
mathematical tractability and because we did not find enough data to develop
an empirical-based model, it is consistent with empirical observations to
some level. In particular, under Assumption \ref{ass:cancelation}, the
cancellation rate is decreasing with respect to the relative price $i\delta $
as is reported in \citet{Gould&etal_2012} and the references therein.
%%Besides, as we shall see, under Assumption 1, the cancelation rate of each limit order is proportional to its chance of being executed.%%
To see why $\alpha (\cdot ;a,b)$ is decreasing, let's assume that $p(i\delta
;a,b)\leq A\delta $ for some constant $A>0$. Since for any fixed $h>0$ and $%
x>0$ small enough, $x/(\log (a+x)-\log (a))\approx a$, we have
\begin{equation*}
\alpha (i\delta ;a,b)\approx \frac{\lambda }{c_{p}}(1-\sum_{j\leq
i-1}p(j\delta ;a,b))
\end{equation*}%
when $\delta $ is small enough and hence it is decreasing in $i\delta $.

Moreover, when $c_{p}=1$, $\theta (\cdot ;a,b)=1-\sum_{j\leq i-1}p(j\delta
;a,b)$. In this case, $\theta (\cdot ;a,b)$ is approximately proportional to
$\alpha (\cdot ;a,b)$, which implies that the impatience level of a standing
limit order at position $i\delta $ (namely $\alpha \left( i\delta
;a,b\right) $) is proportional to its rate of execution as observed by the
arriving market orders (namely $\mu \theta \left( i\delta ;a,b\right) $).
So, the probability that a given limit order in equilibrium at position $%
i\delta $ gets executed before cancellation is equal to $\mu \theta \left(
i\delta ;a,b\right) /(\alpha \left( i\delta ;a,b\right) +\mu \theta \left(
i\delta ;a,b\right) )\approx \mu /(\lambda +\mu )$. Consequently, in this
sense all limit orders have roughly the same probability of execution in
equilibrium.

\subsubsection{Connecting Distribution of Price Increments and LOB's
Distributions in Other Regimes\label{Sec_Con_Others}}

We close this section by briefly discussing another asymptotic regime.
Suppose that one assumes that the cancellation rate per order at relative
price $i\delta $ equals $\alpha \left( i\delta ,a,b\right) =\alpha $
(constant) -- see for example \citet{Cont&Stoikov&Talreja_2010}. Then one
can check that the stationary distribution of the multi-class queues
becomes,
\begin{equation}
\theta \left( i\delta ,a,b\right) =\left( \sum_{l=0}^{\infty }\prod_{k=1}^{l}%
\frac{\lambda F\left( i\delta ,a,b\right) }{\mu +k\alpha }\right) ^{-1}.
\label{eq:other regime}
\end{equation}%
So, if $\lambda _{n},\mu _{n}\rightarrow \infty $ in such a way that $%
\lambda _{n}/\mu _{n}\rightarrow c^{\prime }$ as $n\rightarrow \infty $ and $%
\alpha _{n}/\lambda _{n}\rightarrow 0$, then we obtain
\begin{equation*}
\theta \left( i\delta ,a,b\right) \approx \bar{F}\left( i\delta ,a,b\right)
/c^{\prime },
\end{equation*}%
and arrive at the same conclusion as in Proposition \ref{Result_Hazard_Rate}
which $c_{p}=1$. We therefore believe that the sort of relationship that we
have exposed via Proposition \ref{Result_Hazard_Rate} between the return
distribution and the distribution of orders in the book might be relatively
robust. Under the assumption that $\lambda _{n}/\mu _{n}\rightarrow
c^{\prime }$, there is no stochastic averaging principle such as discussed
in Theorem \ref{THM stochastic averaging}. However, one can obtain a
limiting price process with price increment distributed as (\ref{eq:other
regime}) by observing the LOB and the price in suitably chosen discrete time
intervals.

\section{Continuous Time Model\label{Section_CT_Model}}

We write $\bar{s}(t)=\bar{a}(t)-\bar{b}(t)$ to denote the bid-ask spread
per-trade at time $t$. We shall develop a stochastic model for the
price-spread dynamics in longer time scale (order of several minutes or
more). The model will be a jump-diffusion limit of the discrete price-spread
processes as given in Section \ref{Section_Building_Block}.

We will now introduce the distribution of relative prices in the LOB, $%
p(\cdot ;\bar{a}(t_{k}),\bar{b}(t_{k}))$. We shall impose our assumptions
directly on $\theta (\cdot ;\bar{a}(t_{k}),\bar{b}(t_{k}))$ because we can
go back and forth between $\theta (\cdot ;\bar{a}(t_{k}),\bar{b}(t_{k}))$
and $p(\cdot ;\bar{a}(t_{k}),\bar{b}(t_{k}))$ directly via (\ref%
{eq:expression_theta}). We shall consider a sequence of limit order books
indexed by $n$ and their ask-bid (per-trade) price process $\{(\bar{a}%
^{n}(\cdot ),\bar{b}^{n}(\cdot ))\}$. The dynamic of $(\bar{a}^{n}(\cdot ),%
\bar{b}^{n}(\cdot ))$ is characterized by the arrival rate of market orders
and the price increments. In turn, the price increments will be defined in
terms of auxiliary (spread-dependent) random variables denoted by $\Delta
_{a}^{n}\left( \bar{s}^{n}(t_{k})\right) $ for the ask price process and $%
\Delta _{b}^{n}\left( \bar{s}^{n}(t_{k})\right) $ for the buy price process.
For simplicity in the notation, we often write $\Delta _{a}^{n}\left(
t_{k}\right) $ instead of $\Delta _{a}^{n}\left( \bar{s}^{n}(t_{k})\right) $
(similarly for $\Delta _{b}^{n}\left( t_{k}\right) $). We will assume that
both $\Delta _{a}^{n}\left( t_{k}\right) $ and $\Delta _{b}^{n}\left(
t_{k}\right) $ have the same distribution given the spread-per trade, so we
simply provide the description for $\Delta _{a}^{n}\left( t_{k}\right) $ in
our following assumption which is motivated by the Empirical Observation 3.

\begin{assumption}
\label{ass:price return} \textbf{(Price return distribution) }First define,%
\begin{equation}
\Delta _{a}^{n}\left( \bar{s}^{n}(t_{k})\right) =(1-I_{k}^{a})\cdot
(-1)^{R_{k}^{a}}[U_{k}^{a}/(\sqrt{n}\delta )]\delta +I_{k}^{a}[\bar{s}%
^{n}(t_{k})V_{k}^{a}/(2\delta )]\delta  \label{price return}
\end{equation}%
where:

i) $I_{k}^{a}$ is Bernoulli with $P(I_{k}^{a}=1)=q$ for some $q>0$,

ii) $U_{k}^{a}$ is a random variable with support on $[0,\xi ]$ for $\xi \in
(0,\infty )$.

iii) $R_{k}^{a}$ is Bernoulli with $P(R_{k}^{a}=1)=(1+2\beta/\sqrt{n})/2$
for some $\beta >0$,

iv) $V_{k}^{a}$ is a continuous random variable so that $P(V_{k}^{a}\geq
-1)=1$.

v) the random variables $I_{k}^{a}$, $U_{k}^{a}$, $R_{k}^{a}$ and $V_{k}^{a}$
are independent of each other (independence is assumed to hold across $k$
and for the superindices $a,b$).

Then we let
\begin{equation}
\begin{cases}
\bar{a}^{n}(t_{k+1})=\bar{a}^{n}(t_{k})+\Delta _{a}^{n}\left( \bar{s}%
(t_{k})\right) \vee ([-\bar{s}^{n}(t_{k})/(2\delta )]\delta ), \\
\bar{b}^{n}(t_{k+1})=\bar{b}^{n}(t_{k})-\Delta _{b}^{n}\left( \bar{s}%
(t_{k})\right) \vee ([-\bar{s}^{n}(t_{k})/(2\delta )]\delta ),%
\end{cases}
\label{ask bid dynamics}
\end{equation}%
and this is equivalent to assuming%
\begin{equation*}
\theta (i\delta ,\bar{a}^{n}(t_{k}),\bar{b}^{n}(t_{k}))=P\left( \Delta
_{a}^{n}\left( \bar{s}(t_{k})\right) \vee ([-\bar{s}^{n}(t_{k})/(2\delta
)]\delta )\geq i\delta \right) .
\end{equation*}
\end{assumption}

\textbf{Remarks:}

1. The first term $(1-I_{k}^{a})\cdot (-1)^{R_{k}^{a}}[U_{k}^{a}/(\sqrt{n}%
\delta )]\delta $ captures limit orders tend to cluster close to their
respective best bid or ask prices; the parameter $(1-q)\in \left( 0,1\right)
$ can represents the proportion of orders that are concentrated around the
best bid or ask price. Since, as observed earlier, these correspond to a
substantial proportion of the total number of orders placed, we might choose
$q\approx 0$.

2. The second term $I_{k}^{a}[\bar{s}^{n}(t_{k})V_{k}^{a}/(2\delta )]$
captures the limit orders that are put far away from the current bid or ask
price. In Section \ref{Section_Empirical_Validation}, we shall choose $%
V_{k}^{a}$ to have a density $f_{V}$ with a power-law decaying tails which
are consistent with empirical observations (see Empirical Observation 3). We
also postulate a multiplicative dependence on $\bar{s}^{n}(t_{k})$ to
capture the positive correlation between size of spread and variability in
return distribution as reported in (\citet{Bouchard&etal_2004}).

3. Recall that the most aggressive price ticks that are allowed in our
pre-limit model assumptions are at the mid price; this results in the cap $%
\vee ([-\bar{s}^{n}(t_{k})/(2\delta )]\delta )$ appearing in (\ref{ask bid
dynamics}), which consequently yields $\bar{a}^{n},\bar{b}^{n}$.

4. The asymmetry in the distribution of $R_{k}^{a}$ allows us to introduce a
drift term in the spread, which will be useful to induce the existence of
steady-state distributions. We will validate certain features of the
steady-state distribution of our model vis-a-vis statistical evidence in
Section \ref{Section_Empirical_Validation}.

In addition, we impose the following assumptions on time and space scalings,
which are consistent with Empirical Observations 1 and 2. In order to carry
out a heavy traffic approximation, we consider a sequence of LOB systems
indexed by $n\in \mathbb{Z}^{+}$, such that in the $n$-th system:

\begin{assumption}
\label{ass:scaling} \textbf{(Time and Space Scale)}

\begin{enumerate}
\item The arrival rate of market orders on each side $\mu _{n}=n\mu $;

\item Tick size $\delta _{n}\longrightarrow 0$ so that either $\delta
_{n}=o\left( n^{-1/2}\right) $ or $\delta _{n}=n^{-1/2}$.

\item We assume that $q_{n}=\gamma /n$ for some $\gamma >0$.

%\item The sequence $\xi _{n}$ as defined in Theorem \ref{THM stochastic
%averaging} is given by $\xi _{n}=n^{2}$.
\end{enumerate}
\end{assumption}

In order to explain our scaling, note that the number of jumps,
corresponding to the component involving $V_{k}^{a}$ in (\ref{price return})
is Poisson with rate $\gamma \mu $ so Condition 1 of Assumption \ref%
{ass:scaling} helps us capture jump effects in the limit. The scaling that
we consider implies the existence of two types of arriving limit orders, one
type that arrives more frequently than the other, see Assumption \ref%
{ass:scaling}, part 3. in connection to Assumption \ref{ass:price return},\
part i). This scaling feature, together with the fact that the probability
of an order being executed is roughly constant across the book\ (as
discussed at the end of Section \ref{Section_Conect_LOB_Price}) induces a
much higher cancellation rate close to the spread, which is consistent with
empirical findings.

Now we are ready to state our result on the spread and price dynamics
informed by the limit order book.

\begin{thm}
\label{Thm_Main}For the $n$-th system, let $\bar{s}^{n}(t)=\bar{a}^{n}(t)-%
\bar{b}^{n}(t)$ be the spread process and $\bar{m}^{n}(t)=\bar{a}^{n}(t)+%
\bar{b}^{n}(t)$ be twice of the mean price. Suppose $(\bar{s}^{n}(0),\bar{m}%
^{n}(0))=(s_{0},m_{0})$. Then, under Assumptions 1-4, the pair of processes $%
(\bar{s}^{n},\bar{m}^{n})\in D([0,\infty ),\mathbb{R}^{+}\times \mathbb{R})$
converges weakly to $(\bar{s},\bar{m})\in D([0,\infty ),\mathbb{R}^{+}\times
\mathbb{R})$ with $(\bar{s}(0),\bar{m}(0))=(s_{0},m_{0})$ such that
\begin{equation}
\begin{cases}
d\bar{s}(t) & =- \eta dt+dW_{a}(t)+dW_{b}\left( t\right) +\bar{s}\left(
t_{-}\right) dJ_{1}(t)/2+\bar{s}\left( t_{-}\right) dJ_{2}(t)/2+dL(t), \\
d\bar{m}(t) & =dW_{a}(t)-dW_{b}\left( t\right) +\bar{s}\left( t_{-}\right)
dZ_{1}(t)/2-\bar{s}\left( t_{-}\right) dZ_{2}(t)/2.%
\end{cases}
\label{jump diffusion limit}
\end{equation}%
Here,

\begin{enumerate}
\item $\eta=2\mu \beta E\left( \left[U_{1}^{a}\right]\right)=2\mu \beta
E\left(\left[ U_{1}^{b}\right]\right)$ if $\delta_n=n^{-1/2}$, and $%
\eta=2\mu \beta E\left( U_{1}^{a}\right)=2\mu \beta E\left( U_{1}^{b}\right)$
for $\delta_n=o(n^{-1/2})$.

\item $W_{a}$ and $W_{b}$ are two independent Brownian motions, each with
zero mean and variance rate $\sigma ^{2}=\mu E(\left[ U_{j}^{a}\right]
^{2})= $ $\mu E(\left[ U_{j}^{b}\right] ^{2})$ if $\delta _{n}=n^{-1/2}$,
and $\sigma ^{2}=\mu E((U_{j}^{a})^{2})=\mu E((U_{j}^{b})^{2})$ for $\delta
_{n}=o(n^{-1/2})$.

\item $J_{1}$ and $J_{2}$ are two i.i.d. compound Poisson processes with
constant jump intensity $\gamma \mu $ and the jump density distribution
given by the density of $V_{1}^{a}$.

\item $\bar{s}(t)\geq 0$ and $L(t)$ satisfies: $L(t)=0$, $dL(t)\geq 0$ and $%
\bar{s}(t)dL(t)=0$ for all $t\geq 0$.
\end{enumerate}
\end{thm}

\bigskip

\section{Simulation Results\label{Section_Empirical_Validation}}

We simulate the pair of the spread and mid-price processes according to
their asymptotic approximation $(\bar{s}(\cdot ),M(\cdot ))$ as given by (%
\ref{jump diffusion limit}) under different parameters. We use the
distribution%
\begin{equation*}
f_{U}(x)=r/\xi ,\text{ for }x\in (0,\xi ]\text{ and }P(U=0)=1-r,
\end{equation*}%
and
\begin{equation*}
f_{V}(x)=\frac{(u-1)(\rho +x)^{-u}}{2(\rho ^{1-u}-(c+\rho )^{1-u})}I\left(
x\in \left( 0,c\right) \right) +\frac{(u-1)(\rho -x)^{-u}}{2(\rho
^{1-u}-(1+\rho )^{1-u})}I\left( x\in \left( -1,0\right) \right) ,
\end{equation*}%
for $u\in (1,3]$, and $\rho ,c>0$. We have chosen $f_{V}(\cdot )$ so that in
the pre-limit, the distribution of the orders inside the order book yields a
power-law tail (when $i>0$ is big) so that for fixed $a,b$,
\begin{equation*}
p(i\delta ;a,b)\propto \frac{1}{(c_{2}+i\delta )^{u}},
\end{equation*}%
for some $c_{2}>0$. This choice is justified in view of the following
empirical observation.

\textbf{Empirical Observation 3: Distribution of limit orders inside the
order book. }

Power-law decaying tails in the distribution of the relative prices of
incoming limit orders inside the book have been reported in several
empirical studies on order books in different financial markets (see for
instance \citet{Bouchard&etal_2002}, \citet{Potters&Bouchard_2003} and %
\citet{Zovko&Farmer_2002}). Market data suggest that although incoming limit
orders concentrate around the bid or ask price (according to %
\citet{Bouchard&etal_2002}, half of the limit orders have relative tick
price $-1,0$ and $1$), they spread widely on the order book and the tail of
the relative price, either buy or sell, can be well approximated by a
power-law with some power index $u>0$ (i.e. the proportion of orders at $i$
ticks away from the best quote is proportional to $1/\left( c_{1}+i\right)
^{u}$ for some $c_{1}>0$). The index $u$ varies among different financial
markets as reported in \citet{Bouchard&etal_2002}, %
\citet{Potters&Bouchard_2003} and \citet{Zovko&Farmer_2002} with values $%
u\in (1,3]$. It is also observed that the relative price distributions are
basically symmetric on the sell and buy sides. Moreover, empirical
observations also show that a substantial part of the limit sell (and buy)
orders is clustered close to the ask (and bid) price, as is captured by the
first term involving $U$ (see Remark 1 under Assumption \ref{ass:price
return}).

\bigskip

We are choosing the parametric family of our price-return distribution
directly to match empirical features of the distribution of orders in the
book, so here implicitly we are assuming $c_{p}=1$. This parameter can be
adjusted to better reflect tail behavior of the empirical price-return
distribution.

\bigskip

We proceeded to simulate the spread and mid-price processes according to
their asymptotic approximation $(\bar{s}(\cdot ),M(\cdot ))$ as given by (%
\ref{jump diffusion limit}) under different parameters. We then compute the
stationary distribution of the spread and the volatility process of the
mid-price return from the simulation data. The computation results show that
the joint jump-diffusion dynamics of the spread and mid-price (\ref{jump
diffusion limit}) derived from our LOB model can capture several stylized
features in real spread and price data as reported in \citet{Wyart&etal_2008}%
, \citet{Gould&etal_2012} and the references therein.\bigskip

\textbf{Stationary Distribution of the Spread: } In \citet{Wyart&etal_2008},
the authors study the spread size immediately before trade from Philippine
Stock Exchange market data and they find that the stationary distribution of
the spread is close to an exponential while in some other markets, the
stationary distribution of the spread admits a power-law. We simulated the
spread process $\bar{s}(\cdot )$ according to (\ref{jump diffusion limit})
and estimate the mean $E[\bar{s}(\infty )]$ and standard deviation $std[\bar{%
s}(\infty )]=\sqrt{E\pi \lbrack {(\bar{s}(\infty )-E[\bar{s}(\infty )])^{2}}]%
}$ of the spread under its stationary distribution. In particular, we
estimate expectations under the stationary distribution by the path-average
of the simulated $\bar{s}(\cdot )$. The results are reported in Table \ref%
{table 1} and show that the mean $E[\bar{s}(\infty )]$ is close to the
standard deviation $std[\bar{s}(\infty )]$.

Figure \ref{Figure 2} compares the empirical distribution of the simulated
spread data of the spread $\bar{s}$ under the parameter set (b) and the
exponential distribution with the same mean. Although the stationary
distribution of $\bar{s}$ is roughly well fitted by the exponential
distribution with the same mean as shown in Figure \ref{Figure 2} (a), its
tail is much heavier than exponential and resembles a power-law tail as
shown in Figure \ref{Figure 2} (b). Intuitively, the limit spread process (%
\ref{jump diffusion limit}) is a reflected Brownian motion with jumps. In
this light, one could argue that the stationary distribution of the spread
could be well approximated by a mixture of an exponential distribution and a
power-law distribution because the reflected Brownian motion admits an
exponential stationary distribution and the jump size follows some power-law
distribution.\newline
\begin{table}[h]
\center
\begin{tabular}{c|ccccc|cc}
\hline
& $u$ & $\beta$ & $\xi$ & $r$ & $\mu\gamma$ & $E[\bar{s}(\infty)]$ & $std[%
\bar{s}(\infty)]$ \\ \hline
(a) & 2.8 & 0.25 & 0.02 & 0.25 & 6.75 & 0.1704 & 0.2068 \\
(b) & 2.3 & 0.25 & 0.02 & 0.25 & 6.75 & 0.1812 & 0.2273 \\
(c) & 2.8 & 0.5 & 0.025 & 0.25 & 6.75 & 0.0957 & 0.1056 \\
(d) & 2.8 & 0.25 & 0.02 & 0.5 & 4.5 & 0.1576 & 0.1663 \\ \hline
\end{tabular}%
\caption{The mean and standard deviation of $\bar{s}$ in stationarity under
different sets of parameters. $\protect\rho =0.02 $ and $\protect\mu=9 $ are
the same for case (a) to (d). }
\label{table 1}
\end{table}
%\begin{figure}
%\caption{Exponential-Like Stationary Distribution of $S$}
%  \begin{tabular}{c c}
%    \subfloat[]{\includegraphics[scale=0.4]{ExpFit1.eps}} %&\subfloat[]{\includegraphics[scale=0.4]{ExpFit2.eps}}\\
%    %\subfloat[]{\includegraphics[scale=0.4]{ExpFit3.eps}}&\subfloat[]{\includegraphics[scale=0.4]{ExpFit4.eps}}\\
%  \end{tabular} \hfill
%  \label{Figure 1}
%\end{figure}

\begin{figure}[tbp]
\center
\begin{tabular}{c}
\subfloat[]{\includegraphics[scale=0.5]{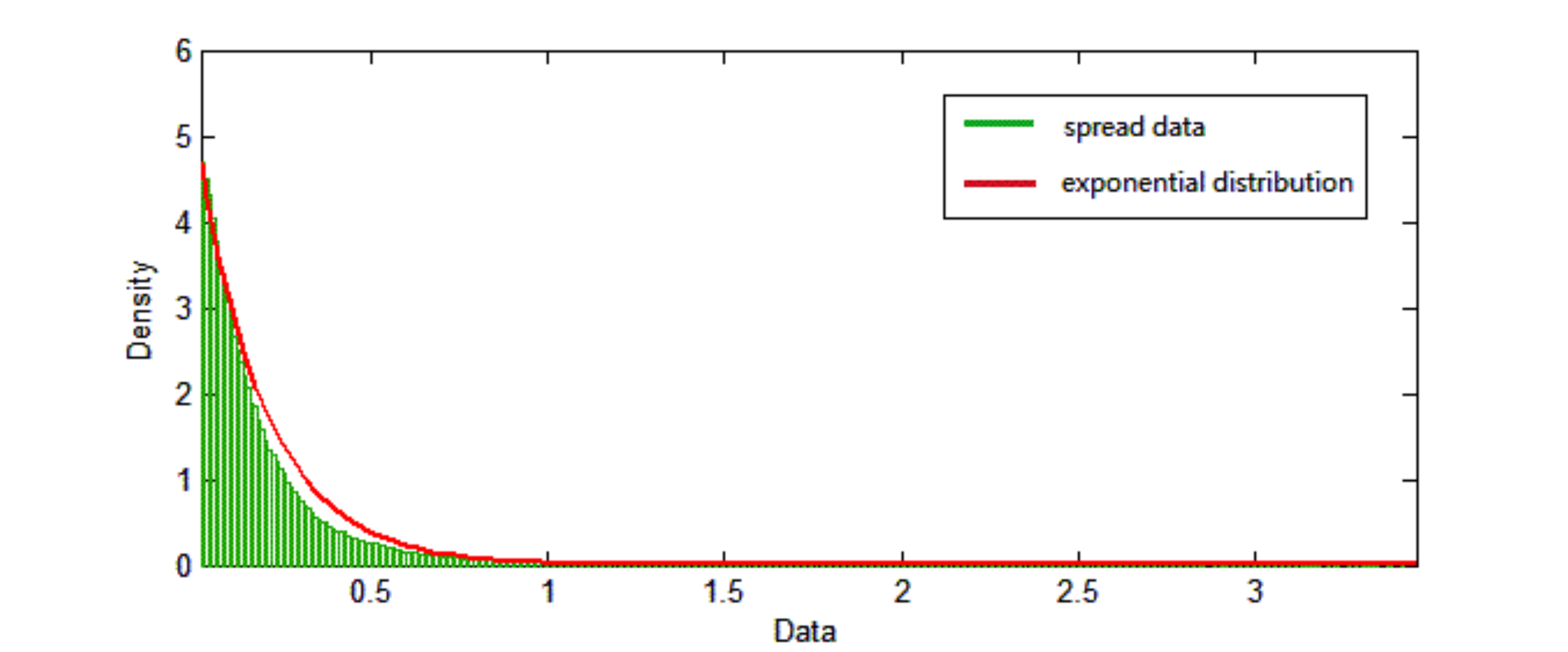}} \\
\subfloat[]{\includegraphics[scale=0.5]{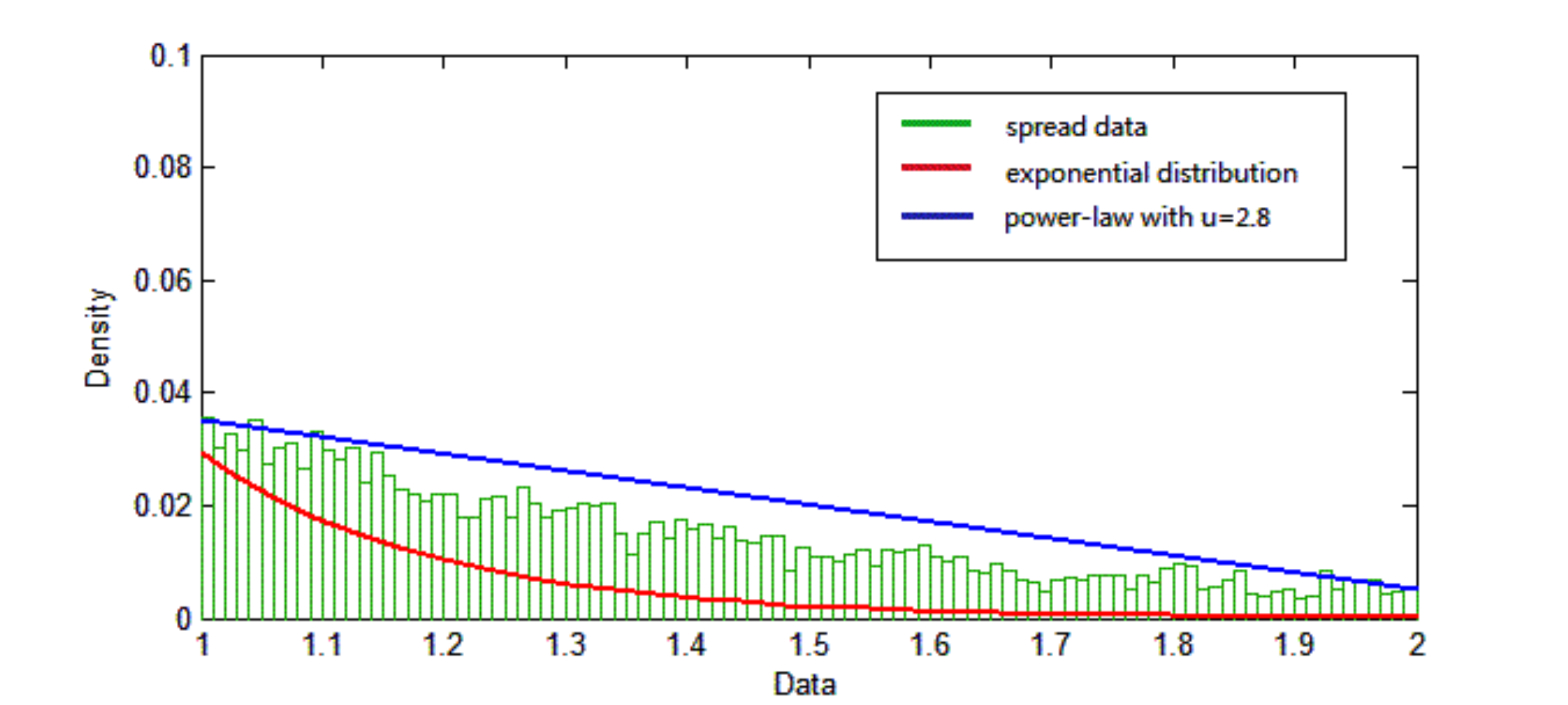}} \\
\end{tabular}%
\caption{Power-law Decaying Tail of $\bar{s}$ }
\label{Figure 2}
\end{figure}
\bigskip \textbf{Correlation between Spread and Volatility:} We study the
relation between the spread and the volatility of the mid-price return
\textit{per trade} as in \citet{Wyart&etal_2008}. In their paper, the
volatility of the mid-price return per trade is computed from empirical data
as $\sigma_1=\sqrt{\sum_{i=1}^N(m_i-m_{i-1})^2/N}$, where $N$ is the number
of trades that has been observed and $m_i$ is the mid-price \textit{per trade%
}. They find a strong linear relationship between $\sigma_1$ and the mean
spread in stationary distribution $E[\bar{s}(\infty)]$. Our model also
capture the linear relationship between the volatility of the mid-price
return and the spread \textit{per trade}. To see this, we first simulate $(%
\bar{s}(\cdot),\bar{m}(\cdot))$ according to (\ref{jump diffusion limit}).
Since $(\bar{s}(\cdot),\bar{m}(\cdot))$ is the limit of the price and spread
\textit{per trade} when the arrival rate of trades $\to \infty$, we estimate
the volatility $\sigma_1$ (up to some constance) of the mid-price return per
trade by
\begin{equation*}
\hat{\sigma}_1^2=\lim_{n\to\infty}\frac{1}{n}\sum_{k=1}^n(\bar{m}(k\Delta t)-%
\bar{m}((k-1)\Delta t))^2/\Delta t,
\end{equation*}
and we choose $\Delta t=0.1$ unit of time. We compute $\hat{\sigma}_1$ as
the path average of the simulated path of $\bar{m}(\cdot)$. We also compute $%
E[\bar{s}(\infty)]$ as the path average taking at every $\Delta t=0.1$ unit
of time interval.
\begin{table}[h]
\center
\begin{tabular}{c|cccccc|cc}
\hline
& $u$ & $\xi$ & $r$ & $\mu$ & $\beta$ & $\mu\gamma$ & $E[\bar{s}(\infty)]$ &
$\hat{\sigma}_1$ \\ \hline
1 & 2.8 & 0.08 & 0.25 & 12 & 0.25 & 9 & 0.1704 & 0.0822 \\
2 & 2.8 & 0.4 & 0.25 & 12 & 0.25 & 9 & 0.7934 & 0.4074 \\
3 & 2.8 & 0.8 & 0.25 & 12 & 0.25 & 9 & 1.5862 & 0.8169 \\
3 & 2.3 & 0.08 & 0.25 & 12 & 0.25 & 9 & 0.1812 & 0.0885 \\
4 & 2.3 & 0.08 & 0.5 & 6 & 0.25 & 3 & 0.1696 & 0.0848 \\
5 & 2.3 & 0.08 & 0.75 & 4 & 0.25 & 1 & 0.1559 & 0.0812 \\ \hline
\end{tabular}%
\caption{Estimation of $E[\bar{s}(\infty )]$ and $\hat{\protect\sigma}_1$
under different parameters.}
\label{table 2}
\end{table}

The simulation results reported in Table \ref{table 2} indicate a linear
relation between $E[\bar{s}(\infty )]$ and $\hat{\sigma}_{1}$ that is found
in \citet{Wyart&etal_2008}. Heuristically, without the jump part in (\ref%
{jump diffusion limit}), $\bar{s}(\cdot )$ becomes a one dimensional
reflected Brownian motion with drift $\xi \beta $ and variance coefficient $%
\xi ^{2}r\mu /3$, and the mid-price is simply a Brownian motion with
variance coefficient $\xi ^{2}r\mu /3$. It is known that the stationary
distribution of a reflected Brownian motion is exponential and one can
compute that $\hat{\sigma}_{1}=\xi r\mu /(6\beta )$. Also, in the case of no
jumps, we can clearly evaluate $\hat{\sigma}_{1}=\xi \sqrt{r\mu /3}$.
Therefore, the mean spread and the mean volatility have a linear
relationship of the form $E[\bar{s}(\infty )]=l\times \hat{\sigma}_{1}$ with
$l=\sqrt{r\mu }/(2\beta \sqrt{3})$. In Table \ref{table 2}, we choose
different sets of parameters such that $\sqrt{r\mu }/(2\beta \sqrt{3})\equiv
2$ and one can check that the estimated mean $E[\bar{s}(\infty )]\approx 2%
\hat{\sigma}_{1}$, so the effect of the jumps is actually relatively minor
on the parameter ranges that we explored, for this particular performance
measure.

%We also find a weak positive correlation between spread size and the
%volatility of mid-price return \textit{per trade} over a longer time as
%reported in \citet{Bouchard&etal_2004}.
%of Especially, we first compute from
%the simulated $(\bar{s}(\cdot ),M(\cdot ))$ the time series of mean spread
%the time series the standard deviation of mid-price return in every 10 units
%of times, and then we compute the correlation coefficient of these two time
%series. Table 5 ADD LABEL reports the computation result for different set
%of parameters.\newline

%\begin{table}[tbp]
%\center
%\begin{tabular}{c|cc|c}
%\hline
%& $u$ & $r$ & corr. coef. \\ \hline
%1 & 2.8 & 0.25 & 0.1347 \\
%2 & 2.8 & 0.5 & 0.1084 \\
%3 & 2.3 & 0.25 & 0.5707 \\
%4 & 2.3 & 0.5 & 0.5100 \\
%5 & 2.3 & 0.75 & 0.1216 \\ \hline
%\end{tabular}%
%\caption{Correlation coefficient between the time series of the spread mean
%and the volatility of mid-price return \emph{per trade}.}
%\end{table}

%\begin{figure}[th]
%\centering
%\includegraphics[scale=0.5]{ExpFit1.pdf} %
%\includegraphics[scale=0.5]{ExpFit2.pdf} %
%\includegraphics[scale=0.5]{ExpFit3.pdf}
%\caption{Exponential Fit of the Stationary Distribution of $S(\cdot)$}
%\end{figure}
\bigskip \textbf{Volatility Clustering: } The jump-diffusion limit (\ref%
{jump diffusion limit}) also captures the volatility clustering feature in
limit order book data as reported in a series of empirical studies (see
Section G.2 in \citet{Gould&etal_2012}). To see this, we measure the
volatility in the mid-price process as the standard deviation of the
mid-price return per 0.1 unit of time over every 10 units long time window.
In detail, we compute a time series $\bar{\sigma}(t)$ from the simulated
mid-price process $\bar{m}$ as
\begin{equation*}
\bar{\sigma}(t)=\sqrt{\frac{1}{99}\sum_{i=1}^{100}(\bar{m}(10t+0.1i)-\frac{1%
}{100}\sum_{j=1}^{100}\bar{m}(10t+0.1j))^2}.
\end{equation*}
To illustrate volatility clustering, in Figure \ref{Figure 3}, we compare
the original time series of the volatility we have computed and its random
permutation. In the original time series, peaks are gathering together while
in the permuted time series, peaks are uniformly distributed along the time.
\begin{figure}[tbp]
\center
{\includegraphics[scale=0.5]{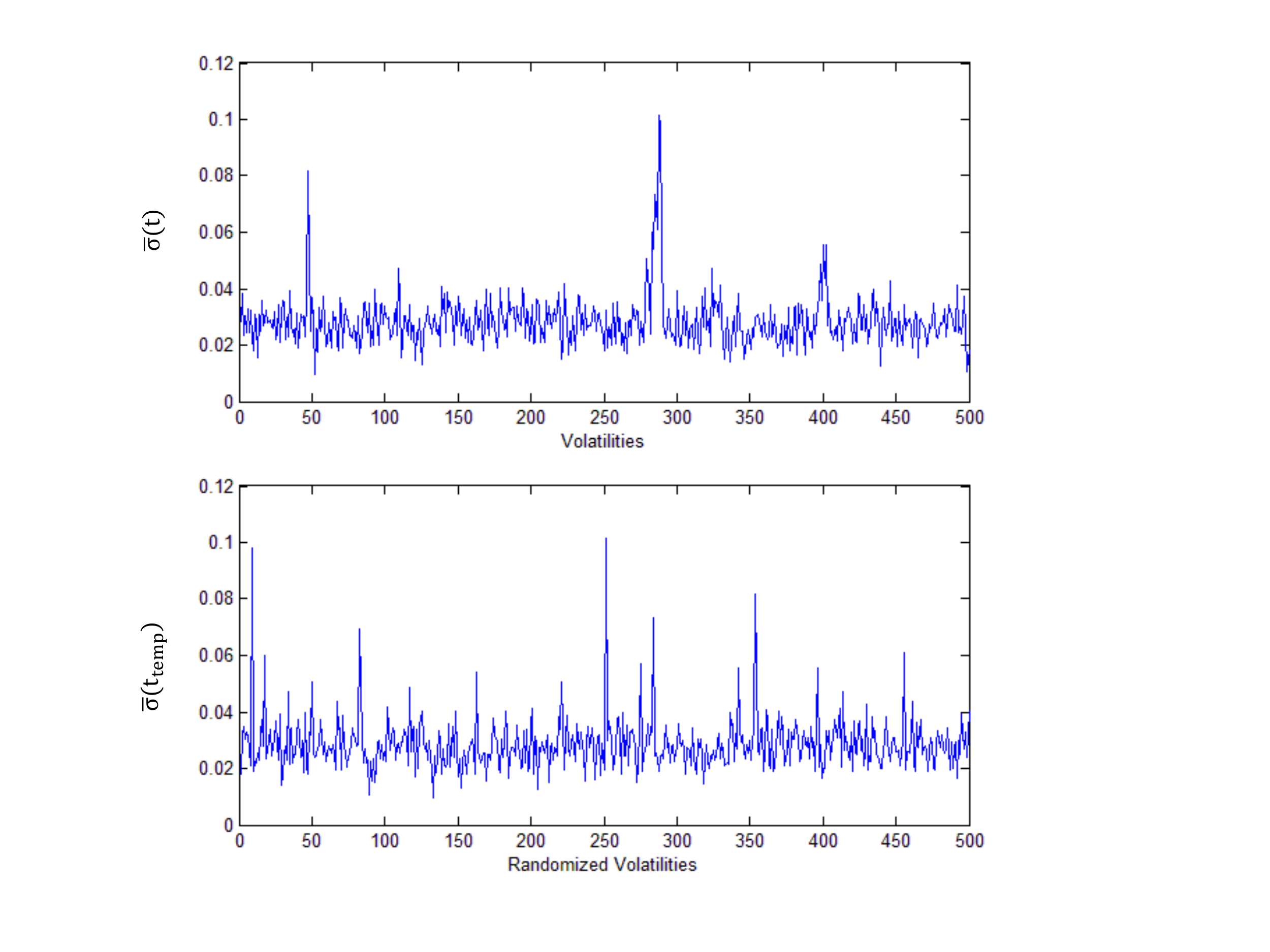}}
\caption{Volatility clustering of realized mid-price returns in simulation
data }
\label{Figure 3}
\end{figure}

\section{Appendix: Technical Proofs\label{Section:App_Technical_Proofs}}

In this section we provide all of the technical proofs in the order in which
the results have been presented.

\subsection{The Proof of Theorem \protect\ref{THM stochastic averaging}\label%
{App_Sect_Proof_THM stochastic averaging}}

\begin{proof}[Proof of Theorem \protect\ref{THM stochastic averaging}]
In this proof, we will follow the notations as used in \citet{Kurtz_1992}.
Define $X_{n}(\cdot )=(\bar{a}_{n}(\cdot ),\bar{b}_{n}(\cdot ))\in \mathbb{Z}%
^{2}$. Define $Y_{n}(\cdot )=(Y_{n}^{i}(\cdot ):i\in \mathbb{Z})\in \mathbb{Z%
}^{\infty }$, where $Y_{n}^{i}(t)$ equals the number of limit sell orders
(or minus the number of limit buy orders) on price tick $i$ in the $n$-th
LOB system at time $t$. We denote by $E_{1}$ and $E_{2}$ respectively the
space of $\mathbb{Z}^{2}$ and $\mathbb{Z}^{\infty }$ that are endowed with
the discrete topology. Then, $\{(X_{n}(\cdot ),Y_{n}(\cdot ))\}$ is a
sequence of stochastic processes living in the product space $E_{1}\times
E_{2}$.

Note that for each $n$, $(X_{n}(\cdot ),Y_{n}(\cdot ))$ is a continuous time
Markov Chain with countable number of states. Let $\{\mathfrak{F}_{t}^{n}\}$
be the natural filtration associated with $(X_{n}(\cdot ),Y_{n}(\cdot ))$.
Then, one can check that for any $f\in C(E_{1})$,%
\begin{equation*}
f(X_{n}(t))-\int_{0}^{t}2\mu
(f(a(Y_{n}(s)),b(Y_{n}(s)))-f(X_{n}(s)))ds\triangleq f(X_{n}(t))-\int_{0}^{t}%
\mathcal{A}f(X_{n}(s),Y_{n}(s))ds
\end{equation*}%
is a martingale with respect to $\{\mathfrak{F}_{t}^{n}\}$. Here the
functional $a:\mathbb{R}^{\infty }\rightarrow \mathbb{R}$ is just the ask
price of the LOB at time $s$, more precisely,
\begin{equation*}
a(Y_{n}(s))=\min \{i:Y_{n}^{i}(s)>0\}.
\end{equation*}%
The functional $b(\cdot )$ can be defined in a similar way. One can also
check that for any $g\in C(E_{2})$,
\begin{align*}
& g(Y_{n}(t))-\int_{0}^{t}\sum_{i}\xi _{n}\left[ \lambda p(i;X_{n}(s))\left(
g(Y_{n}(s)+e_{\bar{a}_{n}(s)+i})+g(Y_{n}(s)+e_{\bar{b}%
_{n}(s)-i})-2g(Y_{n}(s))\right) \right. \\
+& \left. \alpha (i;X_{n}(s))\left( Y_{n}^{\bar{a}%
_{n}(s)+i}(s)(g(Y_{n}(s)-e_{\bar{a}_{n}(s)+i})-g(Y_{n}(s)))+Y_{n}^{\bar{b}%
_{n}(s)-i}(s)(g(Y_{n}(s)-e_{\bar{b}_{n}(s)-i})-g(Y_{n}(s)))\right) \right] ds
\\
\triangleq & ~g(Y_{n}(t))-\int_{0}^{t}\xi _{n}\mathcal{B}%
g(X_{n}(s),Y_{n}(s))ds
\end{align*}%
is also a martingale. Due to the regularity condition imposed (\ref%
{Eq:Reg_Cond_Thm1}), $\{X_{n}(\cdot )\}$ is tight. Therefore, each
subsequence of $\{X_{n}(\cdot )\}$ admit a sub-subsequence that converges
weakly to some sub-limit process $X(\cdot )$. Then according to Theorem 2.1
and the subsequent Example 2.4 in \citet{Kurtz_1992}, each sub-limit process
$X(\cdot )=(\hat{a}(\cdot ),\hat{b}(\cdot ))$ is a solution to the
martingale problem
\begin{equation}
f(X(t))-\int_{0}^{t}\int_{E_{2}}\mathcal{A}f(X(s),y)\pi _{X(s)}(dy)ds,
\label{Xt martingale problem}
\end{equation}%
in the sense that the stochastic process defined by (\ref{Xt martingale
problem}) is a martingale. Moreover, in the expression (\ref{Xt martingale
problem}), $\pi _{X(s)}(\cdot )$ is the unique stationary distribution of a
stochastic process $Y\in E_{2}$ which satisfies the martingale problem
\begin{equation*}
g(Y(t))-\int_{0}^{t}\mathcal{B}g(x,Y(u))du.
\end{equation*}%
In our case, $\pi _{x}(\cdot )$ is simply the stationary distribution of the
LOB system under the parameters $\{p(i;x),\alpha (i;x)\}$.

Now we compute that in (\ref{Xt martingale problem}),
\begin{align*}
& \int_{E_{2}}\mathcal{A}f(X(s),y)\pi _{X(s)}(dy) \\
=& \sum_{i,j}\mu (f(\hat{a}(\cdot )+i,\hat{b}(\cdot )+j)-f(\hat{a}(\cdot ),%
\hat{b}(\cdot )))\pi _{X(s)}(\{a(y)-\hat{a}_{n}(s)=i,b(y)-\hat{b}%
_{n}(s)=-j\}) \\
=& \sum_{i,j}\mu (f(\hat{a}(\cdot )+i,\hat{b}(\cdot )+j)-f(\hat{a}(\cdot ),%
\hat{b}(\cdot )))[\theta (i;\hat{a}(t),\hat{b}(t))-\theta (i+1;\hat{a}(s),%
\hat{b}(s))] \\
& \times \lbrack \theta (j;\hat{a}(s),\hat{b}(s))-\theta (j+1;\hat{a}(s),%
\hat{b}(s))].
\end{align*}%
One can check that the martingale problem (\ref{Xt martingale problem}) has
a unique solution $X(\cdot )$, see for instance Chapter 4.4 in %
\citet{Ethier&Kurtz_1986}. In particular, $X(\cdot )=(\hat{a}(\cdot ),\hat{b}%
(\cdot ))$ is equivalent in distribution to a jump process with jump
intensity $2\mu $ and jump size distribution
\begin{align*}
& P\left( \hat{a}(t)-\hat{a}(t-)=i,\hat{b}(t)-\hat{b}(t-)=j|\hat{a}(t-),\hat{%
b}(t-)\right) \\
=~& [\theta (i;\hat{a}(t-),\hat{b}(t-))-\theta (i+1;\hat{a}(t-),\hat{b}%
(t-))]\times \lbrack \theta (j;\hat{a}(t-),\hat{b}(t-))-\theta (j+1;\hat{a}%
(t-),\hat{b}(t-))].
\end{align*}%
Since $\{X_{n}(\cdot )\}$ is tight and each of its convergent subsequence
admits the same limit $X(\cdot )$, we can conclude that $\{X_{n}(\cdot )\}$
weakly converges to $X(\cdot )$.
\end{proof}

\bigskip

\subsection{The Proof of Theorem \protect\ref{Thm_Main}\label%
{App_Sect_Proof_THM_MAIN}}

Now let us describe the roadmap for the proof of Theorem \ref{Thm_Main}. We
first construct some auxiliary process $(\tilde{S}^{n}(\cdot ),\tilde{M}%
^{n}(\cdot ))$ living in the same probability space as the underlying
process $(\bar{s}^{n}(\cdot ),\bar{m}^{n}(\cdot ))$. The auxiliary process
is a convenient Markov process whose generator can be analyzed to conclude
weak convergence to the postulated limiting jump diffusion (\ref{jump
diffusion limit}). The auxiliary process has the same dynamics as the target
process except when it is on the boundary-layer set $[0,2/\sqrt{n}]\times
\mathbb{R}$. We also show the time spent by the two processes on the
boundary-layer is small and as a result their difference caused by their
different dynamics on the boundary is also small. Actually, such difference
is negligible as $n\rightarrow \infty $ and therefore the target process
converges to the same limit process.

First, we define the auxiliary process coupled with the target process in a
path by path fashion. Recall that by Assumption \ref{price return} and (\ref%
{ask bid dynamics}), we can write
\begin{equation*}
\begin{cases}
\bar{s}^{n}(t_{k+1}) & =\bar{s}^{n}(t_{k})+\Delta _{a}^{n}(\bar{s}%
^{n}(t_{k}))\vee ([-\bar{s}^{n}(t_{k})/2\delta]\delta)+\Delta _{b}^{n}(\bar{s%
}^{n}(t_{k}))\vee ([-\bar{s}^{n}(t_{k})/2\delta]\delta), \\
\bar{m}^{n}(t_{k+1}) & =\bar{m}^{n}(t_{k})+\Delta _{a}^{n}((\bar{s}%
^{n}(t_{k})))\vee ([-\bar{s}^{n}(t_{k})/2\delta]\delta)-\Delta _{b}^{n}(\bar{%
s}^{n}(t_{k}))\vee ([-\bar{s}^{n}(t_{k})/2\delta]\delta).%
\end{cases}%
\end{equation*}

Now we define the auxiliary process $(\tilde{S}^{n}(\cdot ),\tilde{M}%
^{n}(\cdot ))$ coupled with $(\bar{s}^{n}(\cdot ),\bar{m}^{n}(\cdot ))$) as
\begin{equation}
\begin{cases}
\tilde{S}^{n}(t_{k+1}) & =\tilde{S}^{n}(t_{k})+(\Delta _{a}^{n}(\tilde{S}%
^{n}(t_{k}))+\Delta _{b}^{n}(\tilde{S}^{n}(t_{k})))\vee (-\tilde{S}%
^{n}(t_{k})), \\
\tilde{M}^{n}(t_{k+1}) & =\tilde{M}^{n}(t_{k})+\Delta _{a}^{n}((\tilde{S}%
^{n}(t_{k})))-\Delta _{b}^{n}(\tilde{S}^{n}(t_{k})),%
\end{cases}%
\end{equation}%
with the initial condition $\tilde{S}^{n}(0)=\bar{s}^{n}(0)$ and $\tilde{M}%
^{n}(0)=\bar{m}^{n}(0)$. %The intuition behind the process $\tilde{S}%
%^{n}\left( \cdot \right) $ is as follows.

Then the main result in this section Theorem \ref{Thm_Main} is an immediate
corollary of the following two propositions.

\begin{proposition}
\label{prop auxiliary converge} The auxiliary process $(\tilde{S}^n(\cdot),%
\tilde{M}^n(\cdot))$ converges weakly to the limit process given by (\ref%
{jump diffusion limit}).
\end{proposition}

\begin{proposition}
\label{prop error to 0} The difference process $\left( \bar{s}^{n}(\cdot )-%
\tilde{S}^{n}(\cdot ),\bar{m}^{n}(\cdot )-\tilde{M}^{n}(\cdot )\right) $
converges weakly to $(0,0)$ on $D_{\mathbb{R}^{2}}[0,t]$ for any $t<\infty $.
\end{proposition}

\bigskip

\begin{proof}[Proof of Proposition \protect\ref{prop error to 0}]
For simplicity, we assume that $\xi=1$, otherwise we can divide $\tilde{S}$,
$\tilde{M}$ and $\bar{s}$, $\bar{m}$ by the constant $\xi$. Assume also that
$V_{k}^{a}\leq c$ for some $c\geq 1$. The general case can be dealt with
using truncation because there are only a Poisson number of jumps that arise
in $O\left( 1\right) $ time.

Now, let us first give a bound for the difference $\bar{s}^{n}(\cdot )-%
\tilde{S}^{n}(\cdot )$. For fixed $n$, we define $N(t)=\sum_{\{j:t_{j}\leq
t\}}(I_{j}^{a}+I_{j}^{b})$, intuitively $N(t)$ corresponds to the number of
jumps of the limiting process from time 0 to time $t$. Now we prove by
induction that
\begin{equation}
0\leq \bar{s}^{n}(t_{k})-\tilde{S}^{n}(t_{k})\leq \left(
(1+c)^{N(t_{k})}-1\right) \cdot \frac{2}{\sqrt{n}}.  \label{error S}
\end{equation}%
At $t_{0}=0$, we have $\tilde{S}^{n}(t_{0})=\bar{s}(t_{0})$. Now suppose the
relation (\ref{error S}) holds at time $t_{k-1}$, there are two cases at
time $t_{k}$, case 1): $N(t_{k})=N(t_{k-1})$, and case 2): $%
N(t_{k})>N(t_{k-1})$.

First let us consider the case when $N(t_{k})=N(t_{k-1})$. In this case, we
know that $\Delta _{a}^{n}(\bar{s}^{n}(t_{k-1}))=\Delta _{a}^{n}(\tilde{S}%
(t_{k-1})):=\Delta _{a}^{n}$ is independent of $\bar{s}^{n}(t_{k-1})$ and $%
\tilde{S}^{n}(t_{k-1})$. Also, keep in mind that $\left\vert \Delta
_{a}^{n}\left( t_{k}\right) \right\vert \leq 1/\sqrt{n}$. \ Now we can write
the increment of the difference process
\begin{align}
& (\bar{s}^{n}(t_{k})-\tilde{S}^{n}(t_{k}))-(\bar{s}^{n}(t_{k-1})-\tilde{S}%
^{n}(t_{k-1}))  \notag \\
=~& \Delta _{a}^{n}\vee ([-\bar{s}^{n}(t_{k})/(2\delta )]\delta )+\Delta
_{b}^{n}\vee ([-\bar{s}^{n}(t_{k})/(2\delta )]\delta )-(\Delta
_{a}^{n}+\Delta _{b}^{n})\vee (-\tilde{S}^{n}(t_{k-1})).
\label{error increment}
\end{align}%
Therefore, if $\bar{s}^{n}(t_{k-1})\geq \tilde{S}^{n}(t_{k-1})\geq 2/\sqrt{n}
$, we have
\begin{equation*}
(\bar{s}^{n}(t_{k})-\tilde{S}^{n}(t_{k}))-(\bar{s}^{n}(t_{k-1})-\tilde{S}%
^{n}(t_{k-1}))=\Delta _{a}^{n}+\Delta _{b}^{n}-(\Delta _{a}^{n}+\Delta
_{b}^{n})=0
\end{equation*}%
and as a result $\bar{s}^{n}(t_{k-1})-\tilde{S}^{n}(t_{k-1})=\bar{s}%
^{n}(t_{k})-\tilde{S}^{n}(t_{k})$. If $\bar{s}^{n}(t_{k-1})\geq 2/\sqrt{n}%
\geq \tilde{S}^{n}(t_{k-1})\geq 0$, We have
\begin{align*}
& (\bar{s}^{n}(t_{k})-\tilde{S}^{n}(t_{k}))-(\bar{s}^{n}(t_{k-1})-\tilde{S}%
^{n}(t_{k-1})) \\
=~& \Delta _{a}^{n}+\Delta _{b}^{n}-(\Delta _{a}^{n}+\Delta _{b}^{n})\vee (-%
\tilde{S}^{n}(t_{k-1}))=-(\tilde{S}^{n}(t_{k-1})+\Delta _{a}^{n}+\Delta
_{b}^{n})^{-}\leq 0.
\end{align*}%
Therefore,
\begin{equation*}
\bar{s}^{n}(t_{k-1})-\tilde{S}^{n}(t_{k-1})\geq \bar{s}^{n}(t_{k})-\tilde{S}%
^{n}(t_{k})\geq \frac{2}{\sqrt{n}}-(\tilde{S}^{n}(t_{k-1})+\Delta
_{a}^{n}+\Delta _{b}^{n})\geq 0.
\end{equation*}%
Otherwise, we have $0\leq \tilde{S}^{n}(t_{k-1})\leq \bar{s}%
^{n}(t_{k-1})\leq 2/\sqrt{n}$. In this case, one can check that for any
fixed $\tilde{S}^{n}(t_{k})=\tilde{s}$ and $\bar{s}^{n}(t_{k})=s$, the
increment of the difference process (\ref{error increment}) reaches its
maximum at $\Delta _{a}^{n}=-1//\sqrt{n}$ and $\Delta _{b}^{n}=-\tilde{s}+/%
\sqrt{n}$ and its minimum at $\Delta _{a}^{n}=\Delta _{b}^{n}=-[s/2\delta
]\delta $. Hence,
\begin{equation*}
\tilde{s}-s\leq \left( \bar{s}^{n}(t_{k})-\tilde{S}^{n}(t_{k})\right)
-\left( \bar{s}^{n}(t_{k-1})-\tilde{S}^{n}(t_{k-1})\right) \leq 0\vee (\frac{%
1}{\sqrt{n}}-\frac{s}{2}-\tilde{s}).
\end{equation*}%
Plugging in $\tilde{S}^{n}(t_{k})=\tilde{s}$ and $\bar{s}^{n}(t_{k})=s$, we
have
\begin{equation*}
0\leq \bar{s}^{n}(t_{k})-\tilde{S}^{n}(t_{k})\leq (s-\tilde{s})\vee (\frac{1%
}{\sqrt{n}}+\frac{s}{2})\leq (\bar{s}^{n}(t_{k-1})-\tilde{S}%
^{n}(t_{k-1}))\vee \frac{2}{\sqrt{n}}.
\end{equation*}%
The last inequality holds as $s=\bar{s}^{n}(t_{k-1})\leq 2/\sqrt{n}$. In
summary, we have proved that when $N(t_{k})=N(t_{k-1})$, if the relation (%
\ref{error S}) holds at time $t_{k-1}$, so does it at time $t_{k}$.

Now if $N(t_{k})\geq N(t_{k-1})+1$, intuitively, at least one jump occurs in
$\Delta _{a}^{n}$ and $\Delta _{b}^{n}$. If $I_{k}^{a}=1$ we have
\begin{equation*}
\Delta _{a}^{n}(\tilde{S}^{n}(t_{k-1}))=I_{k}^{a}V_{k}^{a}[\tilde{S}%
^{n}(t_{k-1})/(2\delta )]\delta ,\text{ and }\Delta _{a}^{n}(\bar{s}%
^{n}(t_{k-1}))=I_{k}^{a}V_{k}^{a}[\bar{s}^{n}(t_{k-1})/(2\delta )]\delta .
\end{equation*}%
If in addition $I_{k}^{b}=1$, then
\begin{equation*}
\Delta _{b}^{n}(\tilde{S}^{n}(t_{k-1}))=I_{k}^{b}V_{k}^{b}[\tilde{S}%
^{n}(t_{k-1})/(2\delta )]\delta ,\text{ and }\Delta _{b}^{n}(\bar{s}%
^{n}(t_{k-1}))=I_{k}^{b}V_{k}^{b}[\bar{s}^{n}(t_{k-1})/(2\delta )]\delta ,
\end{equation*}%
and therefore,
\begin{equation*}
\bar{s}^{n}(t_{k})-\tilde{S}^{n}(t_{k})\leq (\bar{s}^{n}(t_{k-1})-\tilde{S}%
^{n}(t_{k-1}))(I_{k}^{a}V_{k}^{a}+I_{k}^{b}V_{k}^{b}+1)/2.
\end{equation*}%
As
\begin{equation*}
0\leq V_{k}^{a}+V_{k}^{b}+1\leq 2c+1\text{ \ and \ }0\leq \bar{s}%
^{n}(t_{k-1})-\tilde{S}^{n}(t_{k-1})\leq ((c+1)^{N(t_{k-1})}-1)\cdot \frac{2%
}{\sqrt{n}},
\end{equation*}
by the induction assumption we have
\begin{equation*}
0\leq \bar{s}^{n}(t_{k})-\tilde{S}^{n}(t_{k})\leq
((c+1)^{N(t_{k-1})}-1)(2c+1)\cdot \frac{2}{\sqrt{n}}\leq
((c+1)^{N(t_{k})}-1)\cdot \frac{2}{\sqrt{n}}.
\end{equation*}%
If $I_{k}^{b}=0$, then following a similar argument as in the case when $%
N(t_{k})=N(t_{k-1})$, we have
\begin{align*}
\bar{s}^{n}(t_{k})-\tilde{S}^{n}(t_{k})& \leq (\bar{s}^{n}(t_{k-1})-\tilde{S}%
^{n}(t_{k-1}))(V_{k}^{a}+1)+(\bar{s}^{n}(t_{k-1})-\tilde{S}%
^{n}(t_{k-1}))\vee \frac{2}{\sqrt{n}} \\
& =((c+1)^{N(t_{k})}-1)\cdot \frac{2}{\sqrt{n}}\text{ ~~ as }c\geq 1.
\end{align*}%
In summary, we have proved the relation (\ref{error S}) of $\tilde{S}%
^{n}(\cdot )$ and $\bar{s}^{n}(\cdot )$ by induction.

Now let us turn to the difference $\bar{m}^{n}(\cdot )-\tilde{M}^{n}(\cdot )$%
. Actually, $\bar{m}^{n}(t)-\tilde{M}^{n}(t)$ can be decomposed into two
parts,
\begin{align*}
& \bar{m}^{n}(t)-\tilde{M}^{n}(t) \\
\leq ~& \sum_{0\leq k\leq \lbrack nt]:N(t_{k+1})=N(t_{k})}\left[ \Delta
_{a}^{n}(t_{k})\vee ([\bar{s}^{n}(t_{k-1})/(2\delta )]\delta )-\Delta
_{b}^{n}(t_{k})\vee ([\bar{s}^{n}(t_{k-1})/(2\delta )]\delta )-(\Delta
_{a}^{n}(t_{k})-\Delta _{b}^{n}(t_{k}))\right] \\
& ~+\sum_{i=1}^{N(t)}([\bar{s}^{n}(t_{k-1})/(2\delta )]\delta -[\tilde{S}%
^{n}(t_{k-1})/(2\delta )]\delta )(I_{k_i}^aV_{k_i}^a+I_{k_i}^bV_{k_i}^b),
\end{align*}%
\newline
where $\{t_{k_i}\}$ are the jump times. We denote the two summation parts as
\begin{equation*}
\bar{m}^{n}(t)-\tilde{M}^{n}(t)= \epsilon_0^{n}(t)+\epsilon_1^n(t).
\end{equation*}
Intuitively, $\epsilon_0^{n}(t)$ is the error corresponding to the diffusion
part when $I^a_k=I^b_k=0$ and $\epsilon_1^n(t)$ is the error corresponding
to the jumps. In the summation part $\epsilon_0^n(t)$, we write $\Delta
_a^{n}(\bar{s}^{n}(t_{k}))=\Delta _a^{n}(\tilde{S}^{n}(t_{k}))=\Delta
_a^{n}(t_{k})$ , because they are independent of $\bar{s}^{n}(t_{k})$ and $%
\tilde{S}^{n}(t_{k})$ when when $I^a_k=I^b_k=0 $.

Following a same induction argument as for $\bar{s}^{n}-\tilde{S}^{n}$, we
can show that the error caused by jumps $\epsilon_1^n(t)$ satisfies that
\begin{equation*}
\epsilon_1^n(t)\leq ((1+2c)^{N(t)}-1)\cdot \frac{2}{\sqrt{n}}.
\end{equation*}%
On the other hand, note that $\epsilon_0^n(t)$ equals
\begin{align*}
& \sum_{0\leq k\leq \lbrack nt]:N(t_{k+1})=N(t_{k})}\left[
\Delta_{a}^{n}(t_{k})\vee ([\bar{s}^{n}(t_{k-1})/(2\delta )]\delta
)-\Delta_{b}^{n}(t_{k})\vee ([\bar{s}^{n}(t_{k-1})/(2\delta )]\delta
)-(\Delta_{a}^{n}(t_{k})-\Delta_{b}^{n}(t_{k}))\right] \\
=& \sum_{0\leq k\leq \lbrack nt]:N(t_{k+1})=N(t_{k})}\left[ \left(
\Delta_{a}^{n}(t_{k})\vee ([\bar{s}^{n}(t_{k-1})/(2\delta )]\delta
)-\Delta_{a}^{n}(t_{k})\right) -\left( \Delta_{b}^{n}(t_{k})\vee ([\bar{s}%
^{n}(t_{k-1})/(2\delta )]\delta )-\Delta_{b}^{n}(t_{k})\right) \right] .
\end{align*}%
Since $\Delta_{a}^{n}(t_{k})$ and $\Delta_{b}^{n}(t_{k})$ are independent
and identically distributed, we have that for any $k\geq 1$
\begin{align*}
E[\epsilon_0^{n}(t_{k})-\epsilon_0^{n}(t_{k-1})|\mathcal{F}_{t_{k}}^{n}]=~&
P(N(t_{k})=N(t_{k-1}))\cdot \left( E\left[ \Delta_{a}^{n}(t_{k})\vee ([\bar{s%
}^{n}(t_{k-1})/(2\delta )]\delta )-\Delta_{a}^{n}(t_{k})\right] \right. \\
& \left. -E\left[ \Delta_{b}^{n}(t_{k})\vee ([\bar{s}^{n}(t_{k-1})/(2\delta
)]\delta )-\Delta_{b}^{n}(t_{k})\right] \right) =0,
\end{align*}%
where $\mathcal{F}_{t_{k}}^{n}$ is the $\sigma $-field generated by $%
\{\Delta_{a}^n(t_{i}),\Delta_{b}^n(t_{i}),\tilde{S}^{n}(t_{i})\}_{i=1}^{k}$.
Therefore, the process $\epsilon_0^{n}(\cdot )$ is a martingale under the
filtration \textbf{$\mathcal{F}^{n}$}. Besides, as $|\Delta_a^n(t_{k})|\leq
1/\sqrt{n}$ when $N(t_{k})=N(t_{k-1})$, we have
\begin{equation*}
|\epsilon_0^{n}(t_{k})-\epsilon_0^{n}(t_{k-1})|\leq \frac{2}{\sqrt{n}}~I(%
\bar{s}^{n}(t_{k-1})<\frac{2}{\sqrt{n}}).
\end{equation*}%
The quadratic variation
\begin{equation*}
[\epsilon_0^{n}](t)~\leq \frac{4}{n}\sum_{i=0}^{[nt]}I(\bar{s}^{n}(t_{i})<%
\frac{2}{\sqrt{n}}).
\end{equation*}%
Recall that we have proved $\bar{s}^{n}(\cdot )\geq \tilde{S}^{n}(\cdot )$,
\begin{equation*}
[\epsilon_0^{n}](t)~\leq \frac{4}{n}\sum_{i=0}^{[nt]}I(\tilde{S}^{n}(t_{i})<%
\frac{2}{\sqrt{n}}).
\end{equation*}%
Since $2/\sqrt{n}\rightarrow 0$, for any $\zeta >0$ we have
\begin{equation*}
\lim_{n\rightarrow \infty }[\epsilon_0^{n}](t)~\leq \lim_{n\rightarrow
\infty }4\int_{0}^{t}I(\tilde{S}^{n}(u)<\zeta )du\leq \lim_{n\rightarrow
\infty }4\int_{0}^{t}f^{\zeta }(\tilde{S}^{n}(u))du,
\end{equation*}%
where $f^{\zeta}(\cdot )$ is a smooth function on $\mathbb{R}^{+}$ and
satisfies $f(x)=1$ for all $0\leq x\leq \zeta $, $0\leq f(x)\leq 1$ for $%
\zeta \leq x\leq 2\zeta $ and $f(x)=0$ for $x> 2\zeta $. (Such function can
be constructed, for instance, by convolution.) Since $f^{\zeta }(\cdot )$ is
bounded and $\tilde{S}^{n}(\cdot )$ converges weakly to the limit process (%
\ref{jump diffusion limit}), we have
\begin{equation*}
\lim_{n\rightarrow \infty }E[[\epsilon_0^{n}](t)]~\leq \lim_{n\rightarrow
\infty }4E[\int_{0}^{t}f^{\zeta }(\tilde{S}^{n}(u))du]=4E[\int_0^t f^{\zeta}(%
\bar{s}(u))du]\leq 4E[\int_{0}^{t}I(\bar{s}(u)\leq 2\zeta )],
\end{equation*}%
As the limit process $\bar{s}(\cdot )$ has the same dynamics as a reflected
Brownian motion except when at the finite time of jumps on $[0,t]$, we have $%
E[\int_{0}^{t}I(\bar{s}(u)\leq 2\zeta )]\rightarrow 0$ as $\zeta \rightarrow
0$. Since $\zeta $ can be arbitrarily small, we conclude that the expected
quadratic variation $E[[\epsilon_0^{n}](t)]\rightarrow 0$ as $n\rightarrow 0
$ for any $t<\infty $. By Doob's Inequality, we have that for all fixed $%
\zeta>0$,
\begin{equation*}
P(\max_{0\leq u\leq t}|\epsilon_0^{n}(u)|>\zeta )\leq \frac{%
E[[\epsilon_0^{n}](t)]}{\zeta ^{2}}\rightarrow 0.
\end{equation*}%
Therefore, $\epsilon_0^{n}(\cdot )$ converges weakly to $x(\cdot )\equiv 0$
in space $D[0,t]$ for all $t<\infty $.

In the end, it is given in Assumption \ref{ass:scaling} that $%
nP(I_k^aV_k^a+I_k^bV_k^b\neq 0)=2\gamma+o(1)$, so the counting process $%
N(\cdot )$ converges to a Poisson process with rate $2\mu\gamma$. Therefore,
for any $t<\infty $
\begin{equation*}
E[\max_{0\leq u\leq t}|((2c+1)^{N(u)}-1)\frac{2}{\sqrt{n}}%
|]=E[((2c+1)^{N(t)}-1)\frac{2}{\sqrt{n}}]=O(\frac{1}{\sqrt{n}}).
\end{equation*}%
As a result, the process $((2c+1)^{N(\cdot )}-1)\frac{2}{\sqrt{n}}$
converges weakly to $x(\cdot )\equiv 0$ in space $D[0,t]$. Recall that we
have proved that $((2c+1)^{N(\cdot )}-1)\frac{2}{\sqrt{n}}$ is an upper
bound of $|\bar{s}^{n}(\cdot )-\tilde{S}^{n}(\cdot )|$ and the `jump part'
of $|\bar{m}^{n}(\cdot )-\tilde{M}^{n}(\cdot )|$. As a consequence, we can
conclude that the difference process $(\bar{s}^{n}(\cdot )-\tilde{S}%
^{n}(\cdot ),\bar{m}^{n}(\cdot )-\tilde{M}^{n}(\cdot ))$ converges weakly to
$(0,0) $ on any compact interval $[0,t]$.
\end{proof}

\bigskip

%The proof of Proposition \ref{prop auxiliary converge} follows the canonical
%method that is based on the convergence of generators of corresponding
%Markov process with a careful investigation on the boundary condition as
%explained in \citet{Ethier&Kurtz_1986}.

\begin{proof}[Proof of Proposition \protect\ref{prop auxiliary converge}]
Define $N_{a}(t)=\sum_{\{j:t_{j}\leq t\}}I_{j}^{a}$, $N_{b}(t)=\sum_{%
\{j:t_{j}\leq t\}}I_{j}^{b}$, and note that $N_{a}\left( \cdot \right) $, $%
N_{b}\left( \cdot \right) $ are two independent Poisson processes with rate $%
\gamma \mu $ each. Next, define $\tilde{S}^{n}(0)=\bar{s}^{n}(0)\geq 0$, and
$\tilde{M}^{n}(0)=\bar{m}^{n}(0),$ and set%
\begin{equation*}
S_{a}^{n}\left( t\right) =\sum_{\{j:t_{j}\leq
t\}}(-1)^{R_{j}^{a}}[U_{j}^{a}/(\sqrt{n}\delta _{n})]\delta _{n},\text{ \ }%
S_{b}^{n}\left( t\right) =\sum_{\{j:t_{j}\leq
t\}}(-1)^{R_{k}^{b}}[U_{k}^{b}/(\sqrt{n}\delta _{n})]\delta _{n}).
\end{equation*}%
We will also define
\begin{eqnarray*}
S^{n}\left( t\right) &=&S_{a}^{n}\left( t\right) +S_{b}^{n}\left( t\right) +%
\bar{s}^{n}(0), \\
M^{n}\left( t\right) &=&S_{a}^{n}\left( t\right) -S_{b}^{n}\left( t\right) +%
\bar{m}^{n}(0),
\end{eqnarray*}%
(so by convention we set $t_{0}=0$ and $S^{n}\left( 0\right) =\bar{s}^{n}(0)$%
). Also, we define
\begin{equation*}
R_{1}^{n}(t)=S^{n}\left( t\right) -\min (S^{n}\left( u\right) :u\leq t,0).
\end{equation*}

Let $A=\inf \{t\geq 0:N_{a}\left( t\right) \geq 1\}$ and $B=\inf \{t\geq
0:N_{b}\left( t\right) \geq 1\}$ be the first arrival times of $N_{a}\left(
\cdot \right) $ and $N_{b}\left( \cdot \right) $, respectively. Since
\begin{equation*}
\tilde{S}^{n}(t_{k})=(\Delta _{a}^{n}(\tilde{S}^{n}(t_{k}))+\Delta _{b}^{n}(%
\tilde{S}^{n}(t_{k}))+\tilde{S}^{n}(t_{k}))^{+},
\end{equation*}%
we have that on $\min (A,B)>t$
\begin{equation*}
\tilde{S}^{n}(t)=R^{n}\left( t\right) .
\end{equation*}%
The strategy proceeds as follows.

Step 1): Show that if $\left( \bar{s}^{n}(0),\bar{m}^{n}(0)\right)
\Rightarrow \left( \bar{s}(0),\bar{m}(0)\right) $, the processes $%
(S_{a}^{n}\left( t\right) ,S_{b}^{n}\left( t\right) :t\geq 0)$ converges
weakly in $D[0,\infty )$ to the process $\left( X\left( t\right) :t\geq
0\right) $ defined via%
\begin{eqnarray*}
X_{1}\left( t\right) &=&\bar{s}(0)-\eta t+W_{a}\left( t\right) +W_{b}\left(
t\right) , \\
X_{2}\left( t\right) &=&\bar{m}(0)+W_{a}\left( t\right) -W_{b}\left(
t\right) .
\end{eqnarray*}

Step 2): Once Step 1) has been executed we can directly apply the continuous
mapping principle to conclude joint weak convergence on $[0,\min (A,B))$ of
the processes%
\begin{eqnarray*}
R^{n}(\cdot ) &\Rightarrow &R\left( \cdot \right) :=X_{1}\left( \cdot
\right) -\min \left( X_{1}\left( u\right) :0\leq u\leq \cdot \right) , \\
M^{n}(\cdot ) &\Rightarrow &X_{2}\left( \cdot \right) .
\end{eqnarray*}

Step 3): By invoking the Skorokhod embedding theorem, we can assume that the
joint weak convergence in Step 2) occurs almost surely. We can add the jump
right at time $\min \left( A,B\right) $ without changing the distribution of
$X_{1}\left( t\right) $ and $\tilde{S}^{n}(t)$ for $t<\min \left( A,B\right)
$. More precisely, define%
\begin{equation*}
D=I\left( A<B\right)R\left( A\right) V^{a}/2 +I\left( B\leq A\right) R\left(
B\right) V^{b}/2 ,
\end{equation*}%
where $V^{b}$ and $V^{a}$ are i.i.d. copies of $V_{k}^{b}$ and $V_{k}^{a}$
respectively, and we also define%
\begin{equation*}
D^{n}=I\left( A<B\right) [R^{n}\left( A\right) V^{a}/(2\delta )]\delta
+I\left( B\leq A\right) [R^{n}\left( B\right) V^{b}/(2\delta )]\delta .
\end{equation*}%
Then put on $t\in \lbrack 0,\min \left( A,B\right) ]$
\begin{eqnarray*}
\tilde{S}^{n}(t) &=&R^{n}\left( t\right) I\left( t<\min \left( A,B\right)
\right) +I\left( t=\min \left( A,B\right) \right) (R^{n}\left( \min \left(
A,B\right) \right) +D^{n}), \\
\bar{s}(t) &=&R\left( t\right) I\left( t<\min \left( A,B\right) \right)
+I\left( t=\min \left( A,B\right) \right) (R\left( \min \left( A,B\right)
\right) +D), \\
\tilde{M}^{n}(t) &=&M^{n}(t)I\left( t<\min \left( A,B\right) \right) \\
&&+I\left( t=\min \left( A,B\right) \right) I\left( A<B\right) [R^{n}\left(
A\right) V^{a}/(2\delta )]\delta \\
&&-I\left( t=\min \left( A,B\right) \right) I\left( B\leq A\right)
[R^{n}\left( B\right) V^{b}/(2\delta )]\delta , \\
\bar{m}(t) &=&X_{2}(t)I\left( t<\min \left( A,B\right) \right) \\
&&+I\left( t=\min \left( A,B\right) \right) I\left( A<B\right) R\left(
A\right) V^{a}/2 \\
&&-I\left( t=\min \left( A,B\right) \right) I\left( B\leq A\right) R\left(
B\right) V^{b}/2\delta .
\end{eqnarray*}%
So, assuming Step 2)\ and using Skorokhod embedding we then conclude that%
\begin{equation*}
\sup_{0\leq t\leq \min \left( A,B\right) }\left\vert \tilde{S}^{n}(t)-\bar{s}%
(t)\right\vert +\sup_{0\leq t\leq \min \left( A,B\right) }\left\vert \tilde{M%
}^{n}(t)-\bar{m}(t)\right\vert \rightarrow 0
\end{equation*}%
almost surely.

Step 4): Finally, note that the convergence extends throughout the interval $%
[0,t]$ by repeatedly applying Steps 1) to 3) given that there are only
finitely many jumps in $[0,t]$. Clearly then this procedure completes the
construction to the solution of the SDE (\ref{jump diffusion limit}).

So, we see that everything rests on the execution of Step 1), and for this
we invoke the martingale central limit theorem (see \citet{Ethier&Kurtz_1986}%
, Theorem 7.1.4). Define
\begin{equation*}
Z_{k}^{a}\left( n\right) =(-1)^{R_{k}^{a}}[U_{k}^{a}/(\sqrt{n}\delta
_{n})]\delta _{n},\text{ \ \ \ }Z_{k}^{b}\left( n\right)
=(-1)^{R_{k}^{b}}[U_{k}^{b}/(\sqrt{n}\delta _{n})]\delta _{n}.
\end{equation*}%
We have that
\begin{eqnarray*}
EZ_{k}^{a}\left( n\right) &=&E([U_{k}^{a}/(\sqrt{n}\delta _{n})]\delta
_{n})(-\beta /(\sqrt{n}))=\beta EZ_{k}^{b}\left( n\right) =
\begin{cases}
\frac{-\beta E\left( U_{1}^{a}\right) }{n}+o(1/n) & \text{ if }\delta_n=o(1/%
\sqrt{n}), \\
\frac{-\beta E\left(\left[U_{1}^{a}\right]\right)}{n}+o\left(1/n\right) &
\text{ if }\delta_n=1/\sqrt{n},%
\end{cases}%
\end{eqnarray*}
and
\begin{eqnarray*}
Var\left( Z_{k}^{a}\left( n\right) \right) &=&E([U_{k}^{a}/(\sqrt{n}\delta
_{n})]^{2}\delta _{n}^{2})-\left( EZ_{k}^{a}\left( n\right) \right)
^{2}=E([U_{k}^{a}/(\sqrt{n}\delta _{n})]^{2}\delta _{n}^{2})-O\left(
1/n\right) .
\end{eqnarray*}%
We write
\begin{equation*}
S_{a}^{n}\left( t\right) =M_{a}^{n}\left( t\right) +\mu tEZ_{k}^{b}\left(
n\right)=M_a^n(t)-\eta t+o(1) ,
\end{equation*}%
where $M_{a}^{n}\left( t\right) $ is a martingale, and we have that%
\begin{equation*}
\sup_{0\leq t\leq T}|M_{n}^{a}\left( t\right) -M_{n}^{a}\left( t_{-}\right)
|\leq (\theta /\sqrt{n}+\delta _{n}),
\end{equation*}%
therefore
\begin{equation*}
E\sup_{0\leq t\leq T}|M_{n}^{a}\left( t\right) -M_{n}^{a}\left( t_{-}\right)
|\text{ }+E\sup_{0\leq t\leq T}|M_{n}^{a}\left( t\right) -M_{n}^{a}\left(
t_{-}\right) |^{2}=o\left( 1\right)
\end{equation*}%
as $n\rightarrow \infty $, which verifies conditions a)\ and b.1) from %
\citet{Ethier&Kurtz_1986},Theorem 7.1.4. Moreover, we have that%
\begin{equation*}
\left[ M_{a}^{n},M_{a}^{n}\right] \left( t\right) =\sum_{j=1}^{N\left(
nt\right) }[U_{j}^{a}/(\sqrt{n}\delta _{n})]^{2}\delta _{n}^{2}\rightarrow
t\sigma =\left\{
\begin{array}{cl}
t\mu E(\left[ U_{j}^{a}\right] ^{2}) & \text{if }\delta _{n}=1/\sqrt{n}, \\
t\mu E(\left( U_{j}^{a}\right) ^{2}) & \text{if }\delta _{n}=o(1/\sqrt{n}).%
\end{array}%
\right.
\end{equation*}%
Furthermore, we have that
\begin{equation*}
E\max_{t\leq T}\left\vert \left[ M_{a}^{n},M_{a}^{n}\right] \left( t\right) -%
\left[ M_{a}^{n},M_{a}^{n}\right] \left( t_{-}\right) \right\vert \leq
(\theta /\sqrt{n}+\delta _{n})^{2}=o\left( 1\right) ,
\end{equation*}%
which corresponds to condition b.2)\ in \citet{Ethier&Kurtz_1986},Theorem
7.1.4. Hence, we conclude that
\begin{equation*}
M_{a}^{n}\left( \cdot \right) \Rightarrow W_{a}\left( \cdot \right)
\end{equation*}%
under the uniform topology on compact sets. A completely analogous strategy
is applicable to conclude $M_{a}^{n}\left( \cdot \right) \Rightarrow
W_{b}\left( \cdot \right) $. The convergence holds jointly due to
independence and therefore we obtain the conclusion required in Step 1). As
indicated earlier, Steps 2) to 4)\ now follow directly.
\end{proof}

\bibliographystyle{plainnat}
\bibliography{reference}

\end{document}